\documentclass[twocolumn,preprintnumbers,amsmath,amssymb,pra,a4paper]{revtex4-1}
\usepackage{amssymb}
\usepackage{amsfonts}
\usepackage{amsmath}
\usepackage{amsthm}
\usepackage[pdftex]{graphicx}

\newcounter{thm} %[counter]
\newtheorem{lemmodq}[thm]{Lemma}
\newtheorem{lemmu}[thm]{Lemma}
\newtheorem{lemholevo}[thm]{Lemma}
 
\begin{document}

\title{Gaussian capacity of the quantum bosonic channel with additive correlated Gaussian noise}
\date{\today}
\author{Joachim Sch\"afer} 
\author{Evgueni Karpov} 
\author{Nicolas J. Cerf} 
\affiliation{QuIC, Ecole Polytechnique, CP 165, Universit\'e Libre de Bruxelles, 1050 Brussels, Belgium}
	
\begin{abstract}
We present an algorithm for calculation of the Gaussian classical capacity of a quantum bosonic memory channel with additive Gaussian noise. The algorithm, restricted to Gaussian input states, is applicable to all channels with noise correlations obeying certain conditions and works in the full input energy domain, beyond previous treatments of this problem. As an illustration, we study the optimal input states and capacity of a quantum memory channel with Gauss-Markov noise [J. Sch\"{a}fer, Phys. Rev. A {\textbf{80}}, 062313 (2009)]. We evaluate the enhancement of the transmission rate when using these optimal entangled input states by comparison with a product coherent-state encoding and find out that such a simple coherent-state encoding achieves not less than 90\% of the capacity.
\end{abstract}

\maketitle

\section{introduction}
A central problem of information theory is to derive the capacity of communication channels, which is the maximal information transmission rate for a given available energy. For quantum channels, which may transmit quantum or classical information, one can define the quantum or classical capacity. The present paper is focused on the latter. Just as in classical information theory, the question of whether or not the classical capacity is additive is of central importance. If the capacity is superadditive, then the asymptotic transmission rate can be higher than the maximal rate achievable for a single use of the channel. Although some quantum channels were proven to be additive \cite{MY04,K03,GGLMSY04}, a counterexample was exhibited, proving that this does not hold in general \cite{H09}. The situation becomes rather different for channels with memory, as then the subsequent uses of the channel get linked. In particular, if the memory is modeled by noise correlations between the uses of the channel, it has been shown that the transmission rate may be enhanced by using entangled input states \cite{KW05,MP02,MPV04,KDC06,KM06,D07,BDM05,CCMR05,GM05,PZM08,LPM09,SDKC09,PLM09,LGM10,LM10}. 

A particular interest has been devoted to Gaussian bosonic channels as they model common physical links such as the optical transmission via free space or optical fibers. An overview of these bosonic channels can be found in \cite{CD94}. Among them the most studied are the additive noise and lossy Gaussian channels (see e.g. \cite{HSH99,GGLMSY04,H05} and Refs. therein). For these channels, it was also shown that correlated noise may lead to superadditivity as the transmission rate may be enhanced by input states with some degree of entanglement \cite{CCMR05,GM05,PZM08,LPM09,SDKC09,PLM09,LGM10,LM10}. We remark that all these studies were restricted to the set of Gaussian input states, so that the derived quantity can be viewed as a ``Gaussian capacity''. However, if the conjecture that Gaussian states minimize the output entropy of Gaussian channels could be proven, then this quantity would turn into the actual capacity. For a particular case, i.e., the memoryless lossy Gaussian channel with vacuum noise, this conjecture was proven to hold \cite{GGLMSY04}. In the present paper, we also restrict ourselves to the set of Gaussian input states, that is, we investigate the Gaussian capacity.

In a recent work, we have evaluated the Gaussian capacity of a Gaussian quantum channel with additive Markov correlated noise, restricted to a certain input energy domain \cite{SDKC09}. Correlated Gaussian noise appears, for instance, in the models of downlink communications between satellites and terrestrial stations \cite{LFH81} and a simple description of such correlations can be provided by an underlying Markov process. We found a \emph{quantum water-filling} solution to the Gaussian capacity, similar to the classical water-filling that appears when considering parallel classical Gaussian channels \cite{Cover}. This similarity is very surprising, as we must take into account that a part of the input energy, in addition to being spent on classical modulation, is spent on the preparation of quantum information carriers (e.g., squeezed states), something which has no classical analog. The notion of quantum water-filling appeared for the first time in the discussion of the capacity of a memoryless phase-dependent Gaussian channel \cite{HSH99}, although there the quantum information carriers were considered part of the channel and only the energy cost of classical modulation was considered, thus making this solution a straightforward analog to the classical one.

In this paper, we present a solution to the Gaussian capacity of the Gaussian additive noise channel in the full input energy domain, where noise correlations are given by stationary (shift-invariant) Gauss processes. We show that the method is applicable even to a larger class of noises. We present an algorithm for a numerical solution of the arising optimization problem based on the method of Lagrange multipliers. Although the algorithm was derived independently, the validity of this method is based on arguments which are essentially equivalent to those presented in \cite{PLM09} for a lossy Gaussian channel where the quantum water-filling solution was also obtained. Using this method, we analyze the Gaussian capacity and the associated optimal input states and encoding as a function of the noise parameters for the special case of Gauss-Markov noise, including the limiting case of maximal noise correlations. In addition, we evaluate the gain from using the optimal input states, which are entangled and therefore may be complicated to produce, with respect to easily generated coherent product states. 

The paper is organized as follows. First, we introduce the notion of classical capacity and Gaussian quantum channels in Sec. \ref{sec:classcap}. Then we discuss the solution to the capacity for the one-mode case in Sec. \ref{sec:optproblem} and for an arbitrary number of modes in Sec. \ref{sec:solarb}. Finally, in Sec. \ref{sec:gaussmarkov} we analyze the capacity, the optimal states, and the gain for Gauss-Markov noise in the full range of the correlation strength. The conclusions are provided in Sec. \ref{sec:conclusions} and several mathematical proofs and definitions are provided in the Appendices.

\section{Classical capacity of quantum Gaussian channels}\label{sec:classcap}

In order to transmit classical information via a quantum channel, one defines an alphabet with letters associated with quantum states $\rho^\mathrm{in}_i$. The quantum channel $T$ is a completely positive, trace-preserving linear map acting on the input ``letter'' states: 
\begin{equation}\label{eq:channel}
  \rho^\mathrm{out}_i=T[\rho^\mathrm{in}_i],
\end{equation} 
resulting in output states $\rho^\mathrm{out}_i$. On average the ``letters'' $\rho^\mathrm{in}_i$ appear in the transmitted messages with \emph{a priori} probabilities $p_i$ so that the overall modulated input state is $\rho^\mathrm{in} = \sum_i p_i \rho^\mathrm{in}_i$. By linearity, the action of $T$ on the overall modulated input reads $T[\rho^\mathrm{in}] = \sum_i{p_i \rho^\mathrm{out}_i} \equiv \overline{\rho}$, where $\overline{\rho}$ is referred to as the overall modulated output. The state $\rho^\mathrm{in}$ is physical only if it has finite energy. Therefore it has to obey the energy constraint
\begin{equation}
  \sum\limits_i{p_i \, \mathrm{Tr}(\rho_i^\mathrm{in} \, \hat{a}^\dagger \hat{a})} \leq \overline n,
  \label{eq:enconstr}
\end{equation}
where $\overline{n}$ is the maximum mean photon number per use of the channel, and $\hat{a}$ and $\hat{a}^\dagger$ are the annihilation and creation operators.

The classical capacity of channel $T$ is the maximal amount of classical bits which can be transmitted per invocation of the channel via quantum states in the limit of an infinite number of channel uses. This quantity can be calculated with the help of the so-called \emph{one-shot} capacity, given by the Holevo bound \cite{H73}
\begin{equation}
  C_1(T) = \max_{\{\rho_i^\mathrm{in},p_i\}}{\; \left\{ S\left(\sum_i{p_i \, T[\rho_i^\mathrm{in}]}\right) - \sum\limits_i{p_i \, S(T[\rho_i^\mathrm{in}])} \right\}},
  \label{eq:oscapacity}
\end{equation}
with the von Neumann entropy $S(\rho) = -\mathrm{Tr}(\rho\log{\rho})$ where $\log$ denotes the logarithm to base 2. The maximum in \eqref{eq:oscapacity} is taken over all ensembles of $\{p_i,\rho_i^\mathrm{in}\}$ of probability distributions $p_i$ and pure input states $\rho_i^\mathrm{in}$, because it was proven in \cite{SW97} that the optimal input states for noisy quantum channels are pure. 

The term ``one-shot'' capacity denotes the maximal amount of information that can be transmitted by a single use of the channel $T$. Furthermore, a number of $n$ consecutive uses of the channel $T$ can be equivalently considered as one use of a parallel $n$-mode channel $T^{(n)}$. Then an upper bound to the capacity of the channel $T$ is given by the limit:
\begin{equation}\label{eq:gencapacity}
  C(T) = \lim_{n \rightarrow \infty} \frac{1}{n} C_1(T^{(n)}).
\end{equation}

It has been shown that the latter is the actual capacity for particular memoryless \cite{SW97,H98,GGLMSY04} and forgetful channels \cite{KW05}, but generally is only an upper bound on the capacity. Here, we evaluate \eqref{eq:gencapacity} and find a specific encoding that realizes this maximal value. For simplicity, we refer to $C(T)$ as capacity in the following. 

Let us now consider an $n$-mode optical channel $T^{(n)}$. In the following, the number of modes of this channel corresponds to the number of monomodal channel uses. Each mode $j$ is associated with the annihilation and creation operators $\hat{a}_j,\hat{a}_j^{\dagger}$, respectively, or equivalently to the quadrature operators $\hat{q}_j = (\hat{a}_j + \hat{a}_j^{\dagger})/\sqrt{2},\hat{p}_j = i(\hat{a}^\dagger_j - \hat{a}_j)/\sqrt{2}$ which obey the canonical commutation relation $[\hat{q}_i,\hat{p}_j] = i\delta_{ij}$, where $\delta_{ij}$ denotes the Kronecker $\delta$. By defining the vector $\hat{R} = (\hat{q}_1,...,\hat{q}_n;\hat{p}_1,...,\hat{p}_n)^\mathrm{T}$, we can express the displacement vector $m$ and covariance matrix $\gamma$ that fully characterize an $n$-mode Gaussian state $\rho$ as 
\begin{equation}\label{eq:covmatrix}
	\begin{split}
	m & = \mathrm{Tr}{[ \, \rho  \hat{R}]},\\
	\gamma & = \mathrm{Tr}{[ (\hat{R} - m) \, \rho \, (\hat{R} - m)^{\ensuremath{\mathsf{T}}} ]} - \frac{1}{2}J,\\
	J & = i\begin{pmatrix} 0 & I\\-I & 0 \end{pmatrix},
	\end{split}
\end{equation}
where $J$ is the symplectic or commutation matrix with the $n \times n$ identity matrix $I$.

We consider a continuous encoding alphabet, where instead of a discrete index we use a complex number. A message of length $n$ is encoded in a $2n$-dimensional real vector\\ $\alpha = (\Re{\{\alpha_1\}},\Re{\{\alpha_2\}},...,\Re{\{\alpha_n\}},\Im{\{\alpha_1\}},...,\Im{\{\alpha_n\}})^{\ensuremath{\mathsf{T}}}$. Physically, this corresponds to a displacement of the $n$-partite Gaussian input state defined by the covariance matrix $\gamma_{\rm in}$ (and zero mean) in the phase space by $\alpha$ and is denoted by $\rho^\mathrm{in}_\alpha$. The Wigner function of $\rho^\mathrm{in}_\alpha$ reads
\begin{equation}\label{eq:wigfun}
  W^\mathrm{in}_{\alpha}(R) = \frac{\exp{[-\frac{1}{2}(R-\sqrt{2}\alpha)^{\ensuremath{\mathsf{T}}} \, \gamma_\mathrm{in}^{-1} \, (R-\sqrt{2}\alpha)]}}  {(2\pi)^{n} \sqrt{\det{(\gamma_\mathrm{in})}}},
\end{equation}
where here $R \in \mathbb{R}^{2n}$ and denotes the coordinates in the phase space. In the following we refer to this state as quantum input. 

At this point, we follow the standard procedure, that is, we only consider Gaussian distributions so that the overall modulated input state is a Gaussian mixture $\rho^{\mathrm{in}}=\int d^{2n}\alpha f(\alpha) \rho^\mathrm{in}_\alpha$, where $d^{2n}\alpha = d\Re{\{\alpha_1\}}d\Im{\{\alpha_1\}}...d\Re{\{\alpha_n\}}d\Im{\{\alpha_n\}}$ with (classical) Gaussian distribution
\[
	f(\alpha) = \frac{\exp{[-\alpha^{\ensuremath{\mathsf{T}}} \, \gamma_\mathrm{mod}^{-1} \, \alpha]}}  {(2\pi)^{n} \sqrt{\det{(\gamma_\mathrm{mod})}}}.
\]
We refer to the covariance matrix $\gamma_\mathrm{mod}$ as classical input or classical modulation. We can set the displacement of the nonmodulated input state and the mean of the classical modulation to 0 without loss of generality, because displacements do not change the entropy $S(\rho)$. Thus, since we restrict the set of ensembles over which $C_1(T)$ in \eqref{eq:oscapacity} is maximized to Gaussian states and distributions, we compute the so-called ``Gaussian'' capacity. This is a lower bound to the capacity $C_1(T)$ (hence, for
$C(T)$ too) which would be the actual capacity if the Gaussian minimum output entropy conjecture could be proven to hold \cite{GLMSY04,L09}.

The action of the channel $T^{(n)}$ on an $n$-mode input state carrying the message $\alpha$ reads as in \cite{CCMR05}
\begin{equation}
  \begin{split}
         & T^{(n)}[\rho^{\mathrm{in}}_\alpha] = \rho^\mathrm{out}_\alpha = \int d^{2n} \beta \, f_\mathrm{env}(\beta) \\
& \times \hat{D}(\beta_n) \otimes ... \otimes \hat{D}(\beta_1) \; \rho^\mathrm{in}_\alpha \; \hat{D}^{\dagger}(\beta_1) \otimes ... \otimes \hat{D}^{\dagger}(\beta_n),
  \end{split}
  \label{eq:changen}
\end{equation} 
with $\beta = (\Re{\{\beta_1\}},...,\Re{\{\beta_n\}},\Im{\{\beta_1\}},...,\Im{\{\beta_n}\})^{\ensuremath{\mathsf{T}}}$ and the displacement operator $\hat{D}(\beta_j) = e^{\beta_j \hat{a}_j^\dagger - \beta_j^{*} \hat{a}_j}$. The displacement is applied according to the (classical) Gaussian distribution of the noise $f_\mathrm{env}(\beta)$ with covariance matrix $\gamma_\mathrm{env}$ (which also will be referred to as ``environment''). If this matrix is not diagonal, then the environment introduces correlations between the successive uses of the channel. These correlations model the memory of the channel.

The covariance matrices of the nonmodulated output state $\rho^\mathrm{out}_\alpha$ and the overall modulated output state $\overline{\rho}$, read, respectively,
\begin{equation}
  \begin{split}
    \gamma_\mathrm{out} & = \gamma_\mathrm{in} + \gamma_\mathrm{env}\\
    \overline{\gamma}   & = \gamma_\mathrm{out} + \gamma_\mathrm{mod}.
  \end{split}
  \label{eq:chancov}
\end{equation}
Note that the covariance matrix does not depend on the displacement as seen in definition \eqref{eq:wigfun}. Therefore, the covariance matrix of $\rho^\mathrm{out}_\alpha$ is the same for all $\alpha$ and the output entropy of all displaced states is identical. Thus, the one-shot capacity \eqref{eq:oscapacity} of this additive channel reduces to
\begin{equation}
  C_1(T^{(n)}) = \sup_{\gamma_\mathrm{in},\gamma_\mathrm{mod}}{\{S(\overline{\rho}) - S(\rho^\mathrm{out}_\alpha)\}}.
  \label{eq:ncapacity}
\end{equation}
We use, in the following, the reduced Holevo $\chi$-quantity that reads
\begin{equation}\label{eq:chi}
	\chi = S(\overline{\rho}) - S(\rho^\mathrm{out}_\alpha).
\end{equation} 

The von Neumann entropy of a Gaussian state $\rho$ is expressed in terms of the symplectic eigenvalues $\nu_j$ of its covariance matrix:
\begin{eqnarray}\label{eq:S}
	S(\rho) & = & \sum\limits_{i=1}^n{g\left(\nu_i-\frac{1}{2}\right)}, \\
	g(x) & = & \left\{ \begin{array}{ll} (x+1)\log{(x+1)} - x\log{x}, &, \,  x > 0\\
											0, &, \, x = 0.
    				   \end{array}\right.\nonumber
	\label{eq:entropy}
\end{eqnarray}
We note that for quantum states the symplectic eigenvalues are always greater or equal to 1/2.
%%%%%%%%%%%%%%%%%%%%%%%%%%%%%
\section{Optimization problem}\label{sec:optproblem}
%%%%%%%%%%%%%%%%%%%%%%%%%%%%%
\subsection{One mode channel}\label{sec:onemode}
%%%%%%%%%%%%%%%%%%%%%%%%%%%%%%%%%%%
In this subsection we consider the case of a single use of the channel $T^{(1)}$ in a similar fashion as in \cite{SDKC09}. We start our discussion with the results obtained in \cite{SDKC09}. However, we present the solution for the whole input energy domain. We consider the noise covariance matrix to have different variances in the quadratures, denoted $\gamma_{\mathrm{env}}^{q,p}$, where we choose without loss of generality $\gamma_{\mathrm{env}}^{q} > \gamma_{\mathrm{env}}^{p}$ and off diagonal terms $\gamma_{\mathrm{env}}^{qp} = 0$. Any $2 \times 2$ covariance matrix can be reduced to this form by a symplectic and orthogonal transformation, which changes neither the entropy of the output state nor the energy constraint. Therefore, by our choice we do not lose generality. As already discussed, we restrict ourselves to the optimization over Gaussian states.

In the following we determine the Gaussian capacity under the following constraints. The first is the condition that $\rho^{\rm in}$ is a pure state which together with the definition \eqref{eq:covmatrix} and the commutation relation imply
\begin{equation}\label{eq:purity}
	\det{\gamma_{\rm in}} = \frac{1}{4}.
\end{equation}
The second is the input energy constraint \eqref{eq:enconstr}, that reads 
\begin{equation}\label{eq:onelambda}
 	\gamma_{\rm in}^q + \gamma_{\rm in}^p + \gamma_{\rm mod}^q + \gamma_{\rm mod}^p = 2\overline{n} + 1 \equiv \lambda,
\end{equation}
where $\gamma_{\rm in}^{q,p}, \gamma_{\rm mod}^{q,p}$ are the diagonal elements of the matrices $\gamma_{\rm in},\gamma_{\rm mod}$ and $\lambda$ will be referred to as ``input energy'' in the following. Furthermore, in order for $\gamma_{\rm in}^{q,p}$ and $\gamma_{\rm mod}^{q,p}$ to be physical they have to be positive. The optimization problem is solved by using the Lagrange multipliers method, with the total Lagrangian being
\begin{equation}\label{eq:Lone}
	\begin{split}
	{\mathcal L} & = g\left( \overline{\nu} - \frac{1}{2} \right)  - \, g\left( \nu_{\rm out} - \frac{1}{2} \right) -\tau\left(\gamma_{\rm in}^{q}\gamma_{\rm in}^{p} - (\gamma_\mathrm{in}^{qp})^2 - \frac{1}{4} \right)\\
	& - \mu\left(\gamma_{\rm in}^{q}+\gamma_{\rm in}^{p} +\gamma_{\rm mod}^{q}+\gamma_{\rm mod}^{p} -\lambda\right),
	\end{split}
\end{equation}
with
\begin{equation}
	\begin{split}
	& \overline{\nu} = \sqrt{\overline{\gamma}^q \overline{\gamma}^p - (\gamma_\mathrm{in}^{qp} + \gamma_\mathrm{mod}^{qp})^2},\\
	& \nu_{\rm out} = \sqrt{ \gamma_{\rm out}^q \gamma_{\rm out}^p - (\gamma_\mathrm{in}^{qp})^2},
	\end{split}
\end{equation}
where $\gamma_\mathrm{in}^{qp}$, $\gamma_\mathrm{mod}^{qp}$ denote the off-diagonal terms of matrices $\gamma_\mathrm{in}, \gamma_\mathrm{mod}$, where $\overline{\gamma}^{q,p}$, $\gamma_{\rm out}^{q,p}$ denote the diagonal elements of $\overline{\gamma}$, $\gamma_{\rm out}$ and $\tau,\mu$ are Lagrange multipliers. In the following we summarize the solution of the system of equations which correspond to the stationary point of ${\mathcal L}$. The details are presented in Appendix \ref{sec:lagrange}. We also prove in Appendix \ref{sec:lagrange} that ${\mathcal L}$ is concave at this stationary point, which implies that the found solution is indeed a local maximum of ${\mathcal L}$. Though we do not have analytic proof that this maximum is global, all our numerical studies confirm this.

\subsubsection{Quantum-waterfilling solution}\label{sec:qwf}
%%%%%%%%%%%%%%%%%%%%%%%%%%%%%%%%%%%
First, we find that the solution implies that the input and modulation covariance matrix cross terms vanish, that is $\gamma_{\rm in}^{qp} = \gamma_{\rm mod}^{qp} = 0$. Therefore, the diagonal elements $\gamma_{\rm in}^{q,p}, \gamma_{\rm mod}^{q,p}$ are the eigenvalues of $\gamma_{\rm in}, \gamma_{\rm mod}$. Then one obtains the \emph{quantum water-filling} condition \cite{SDKC09} (see Appendix \ref{sec:lagrangeqwf} for details)
\begin{equation}\label{eq:oneqwf}
	\gamma_{\rm in}^q + \gamma_{\rm env}^q + \gamma_{\rm mod}^q = \gamma_{\rm in}^p + \gamma_{\rm env}^p + \gamma_{\rm mod}^p,
\end{equation}
which means that the total output energy
\begin{equation}
	\overline{\lambda} \equiv \lambda + \gamma_{\rm env}^q + \gamma_{\rm env}^p
\end{equation}
is uniformly distributed between the quadratures, that is 
\begin{equation}\label{eq:nuwf}
	\overline{\gamma}^q = \overline{\gamma}^p = \frac{\overline{\lambda}}{2} = \overline{\nu}_{\rm wf}.
\end{equation}
We see that this value is also equal to the overall output symplectic eigenvalue which will be referred to as \emph{water-filling level}. The optimal quantum input is found to be determined by the noise variance ratio, i.e.,
\begin{equation}\label{eq:oneoptin}
	\frac{\gamma_{\rm in}^q}{\gamma_{\rm in}^p} = \frac{\gamma_{\rm env}^q}{\gamma_{\rm env}^p}.
\end{equation}
One of the equations links the Lagrange multiplier $\mu$
\begin{equation}\label{eq:oneqwfmu}
	g'\left(\overline{\nu}_{\rm wf}-\frac{1}{2}\right)=2\mu \Rightarrow \overline{\nu}_{\rm wf} = \frac{1}{2}\coth{(\mu \ln{2})}
\end{equation}
to the water-filling level ${\overline{\nu}}_{\rm wf}$, where $g'(x)$ denotes the derivative of $g(x)$ with respect to $x$ and $\ln{(x)}$ is the natural logarithm. In order for the solution to be physical, all eigenvalues in \eqref{eq:oneqwf} have to be positive. This requires an input energy $\lambda$ being above the threshold $\lambda_{\rm thr}$, i.e.,
\begin{equation}\label{eq:onethr}
	\lambda \geq \lambda_{\rm thr} = \sqrt{\frac{\gamma_\mathrm{env}^q}{\gamma_\mathrm{env}^p}} + \gamma_\mathrm{env}^q - \gamma_\mathrm{env}^p.
\end{equation}

The optimal eigenvalues which are the solutions of \eqref{eq:oneqwf}, \eqref{eq:oneoptin} are schematically depicted in Fig.~\ref{fig:monqwf} (a) for a particular noise covariance matrix.
\begin{figure}	 
	\includegraphics[width=0.5\textwidth]{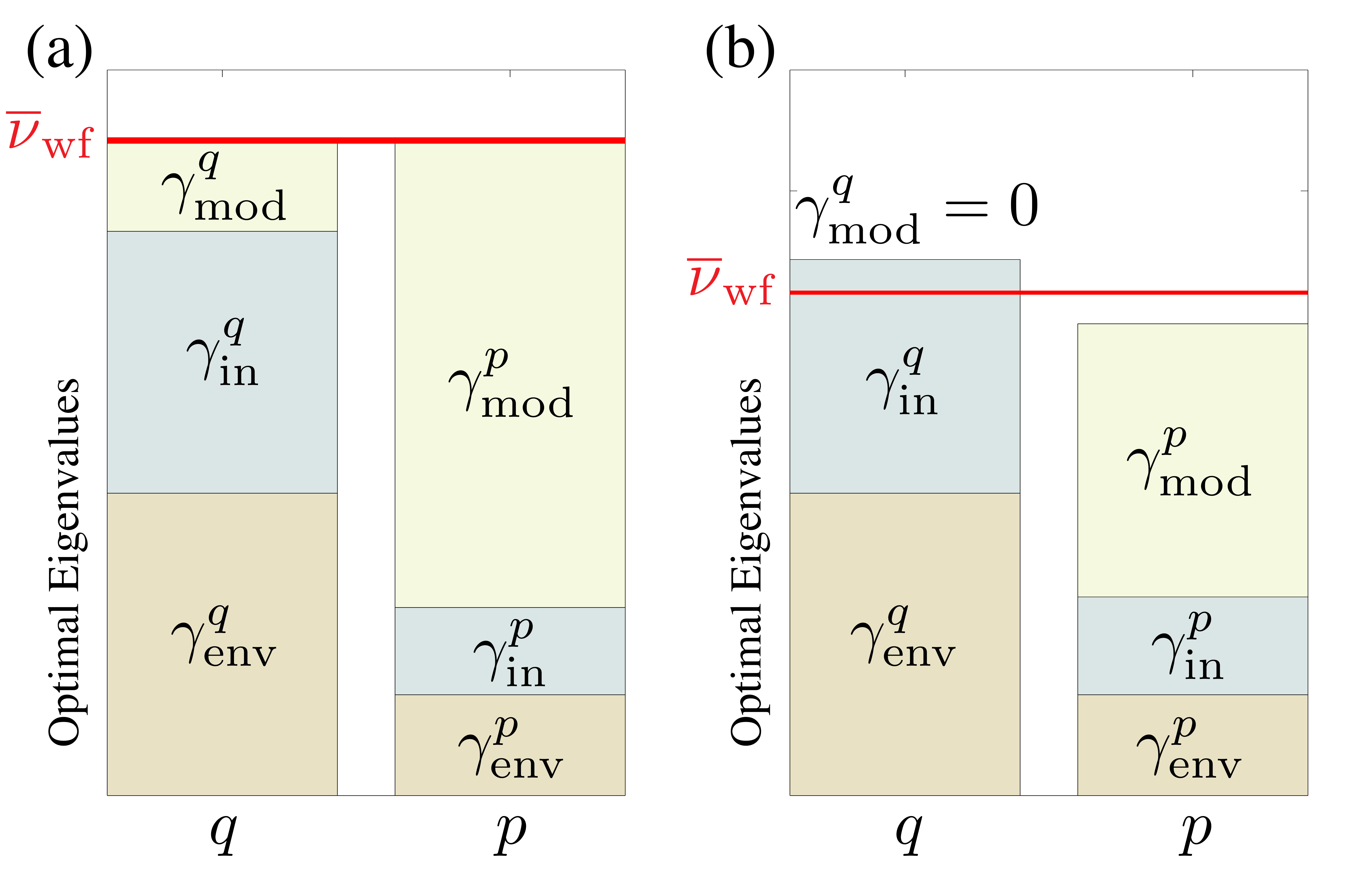}
	\caption{(Color online) Stacked bar plot: optimal eigenvalues of both quadratures. Above and below the threshold the optimal input state is squeezed and the less noisy quadrature is more modulated. (a) $\lambda > \lambda_{\rm thr}$. The optimal eigenvalues are determined by the quantum water-filling solution. ${\overline{\nu}}_{\rm wf}$ denotes the water-filling level. (b) $\lambda < \lambda_{\rm thr}$. The modulation $\gamma_{\rm mod}^q = 0$ since $\gamma_{\rm env}^q > \gamma_{\rm env}^p$. ${\overline{\nu}}_{\rm wf}$ denotes the ``virtual'' water-filling level.}
	\label{fig:monqwf}
\end{figure}
The one-shot Gaussian capacity above the threshold reads
\[
  C_1 = g\left(\overline{n} + \frac{\gamma_{\rm env}^q + \gamma_{\rm env}^p}{2}\right) - g\left(\sqrt{\gamma_{\rm env}^q \, \gamma_{\rm env}^p}\right).
\]
It was shown in \cite{H05} that the one-shot Gaussian capacity is additive for such a tensor product channel and therefore $C_1$ is the Gaussian capacity $C$ above threshold.

%%%%%%%%%%%%%%%%%%%%%%%%%%%%%%%%%%%
\subsubsection{Solution below the threshold}\label{sec:noqwf}
%%%%%%%%%%%%%%%%%%%%%%%%%%%%%%%%%%%
If the input energy $\lambda$ is below the threshold \eqref{eq:onethr}, then the solution of the Eqs. \eqref{eq:oneqwf} and \eqref{eq:oneoptin} and the purity constraint \eqref{eq:purity} might result in negative modulation eigenvalues, which would be unphysical. Indeed, from the solution given by Eqs.~\eqref{eq:oneqwf} and \eqref{eq:oneoptin} one can see that when $\lambda$ decreases the water-filling level $\overline{\nu}_{\rm wf}$ as defined in Eq.~\eqref{eq:nuwf} [see Fig.~\ref{fig:monqwf}~(a)] also decreases and at $\lambda=\lambda_{\rm thr}$ it crosses the level $\gamma_{\rm in}^q~+~\gamma_{\rm env}^q$. For lower $\lambda$, in order to satisfy Eqs.~\eqref{eq:oneqwf} and \eqref{eq:oneoptin} the modulation eigenvalue $\gamma_{\rm mod}^q$ becomes negative. In this case the solution of the optimization problem lays on the border of the valid physical region. As shown in Appendix \ref{sec:lagrangenowf} we have to set $\gamma_{\rm mod}^{q} = 0$ [see Fig.~\ref{fig:monqwf} (b)] and to solve the new optimization problems with a modified ${\mathcal L}$ taking into account this new condition. Here we summarize the results which are presented in details in Appendix \ref{sec:lagrangenowf}. We find that the off-diagonal terms again vanish, i.e., $\gamma_{\rm in}^{qp} = \gamma_{\rm mod}^{qp} = 0$. 

Now, one obtains a transcendental equation that solves the problem below the threshold, i.e.,
\begin{equation}\label{eq:onebt}
		\frac{g'(\overline{\nu} - \frac{1}{2})}{2 \overline{\nu}}\left({\overline{\gamma}}^p - {\overline{\gamma}}^q\right) = \frac{g'(\nu_{\rm out} - \frac{1}{2})}{2 \nu_{\rm out}}\left(\gamma_{\rm out}^p - \frac{\gamma_{\rm in}^p}{\gamma_{\rm in}^q}\gamma_{\rm out}^q\right).
\end{equation}
In Appendix \ref{sec:bounds} we derive from this equation that
\[
	\frac{1}{2} \leq \gamma_{\rm in}^q < \frac{1}{2} \sqrt{\frac{\gamma_{\rm env}^q}{\gamma_{\rm env}^p}}
\]
which means that the more noisier quadrature is always antisqueezed, and therefore the less noisy quadrature is squeezed. We remark that in the limit $\lambda \rightarrow \lambda_{\rm thr}$ the solution of \eqref{eq:onebt} coincides with the water-filling solution \eqref{eq:oneqwf}, \eqref{eq:oneoptin}. The lower bound is reached for $\lambda \rightarrow 1$ and thus corresponds to vanishing modulation eigenvalues (absence of information transmission). We note that Eq.~\eqref{eq:oneqwfmu} which determined the water-filling level reads now
\begin{equation}\label{eq:onebtmu}
	\frac{g'(\overline{\nu} - \frac{1}{2})}{2 \overline{\nu}} {\overline{\gamma}}^q = \mu.
\end{equation}
This equation now determines $\mu$ using the solution of \eqref{eq:onebt}. Although we clearly have no longer a water-filling solution, we can calculate the quantity $\overline{\nu}_{\rm wf}$ following \eqref{eq:oneqwfmu} with $\mu$ determined by \eqref{eq:onebtmu}, but refer to it as a ``virtual'' water-filling level. This quantity will have an important meaning when evaluating the solution for the multimode channel. The reason is that Eqs.~\eqref{eq:oneqwfmu} and \eqref{eq:onebtmu} in the multimode channel problem will govern the distribution of input energy between the channels, because $\mu$ is a monotonically decreasing function of $\lambda$. For $\lambda \geq \lambda_{\rm thr}$ this can be seen from the fact that $\mu$ is a monotonically decreasing function of $\overline{\nu}_{\rm wf}$ [see Eq.~\eqref{eq:oneqwfmu}] and $\overline{\nu}_{\rm wf}$ is a monotonically increasing function of $\lambda$. For $\lambda < \lambda_{\rm thr}$ this is proven in Appendix \ref{sec:mulambda}. This property allows us to relate $\lambda_{\rm thr}$ via Eq.~\eqref{eq:oneqwfmu} to 
\begin{equation}
	\mu_{\mathrm{thr}} = \frac{1}{2} g'\left(\frac{1}{2}(\lambda_{{\rm thr}}+\gamma_{{\rm env}}^{q}+\gamma_{{\rm env}}^{p}) - \frac{1}{2}\right)
	\label{eq:muthr}
\end{equation}
such that if $\lambda \geq \lambda_{\rm thr}$, then the corresponding $\mu \leq \mu_{\rm thr}$. Moreover, for the lowest input energy $\lambda = 1$ we can define using Eq.~\eqref{eq:onebtmu} an upper bound $\mu_0$ for all possible values of $\mu$ that correspond to $\lambda > 1$, that reads
\begin{equation}
	\mu_{0} = \frac{1}{2} g'\left(\sqrt{(\gamma_{{\rm env}}^{q}+\frac{1}{2})(\gamma_{{\rm env}}^{p}+\frac{1}{2})} - \frac{1}{2}\right) \sqrt{\frac{\gamma_{{\rm env}}^{q}+
\frac{1}{2}}{\gamma_{{\rm env}}^{p}+\frac{1}{2}}}.
	\label{eq:mu0}
\end{equation}

In Appendix \ref{sec:bounds} and \ref{sec:mulambda} we draw additional conclusions from Eqs.~\eqref{eq:onebt} and \eqref{eq:onebtmu}. We show that $\overline{\gamma}^p<\overline{\gamma}^q$ and $\overline{\gamma}^p \geq 1/2$. Furthermore, we find that $d \gamma_{\rm in}^q/ d \lambda > 0$ which means that the antisqueezing is increasing with the input energy. Moreover, we find that $d \overline{\gamma}^p/d \lambda > 0$ which implies together with $d \gamma_{\rm in}^q/ d \lambda > 0$ that the modulation in the less noisy quadrature is increasing with $\lambda$. 

The optimal eigenvalues follow from the solution of \eqref{eq:onebt} and are schematically depicted in Fig.~\ref{fig:monqwf} (b) for a particular chosen noise. The Gaussian capacity below the threshold is then calculated by inserting the optimal eigenvalues into Eq.~\eqref{eq:chi}.

%%%%%%%%%%%%%%%%%%%%%%%%%%%%%%%%%%%
\subsection{Multimode channel}\label{sec:multimode}
%%%%%%%%%%%%%%%%%%%%%%%%%%%%%%%%%%%
We consider in this paper channels with noise correlations only between the uses of the channel without $q-p$ correlations. In the equivalent representation of $n$ successive uses of such a one-mode channel by an $n$-mode parallel channel the covariance matrix of the noise has the form
\begin{equation}\label{eq:envcov}
	\gamma_{\rm env} = 
	\begin{pmatrix}
		\gamma_{\rm env}^q & 0\\
		0 & \gamma_{\rm env}^p
	\end{pmatrix}
\end{equation}
where $\gamma_{\rm env}^q$, $\gamma_{\rm env}^p$ are matrices of dimension $n \times n$. The absence of $q$ and $p$ correlations is generally considered to describe a natural noise. If $\gamma_{\rm env}^q$ and $\gamma_{\rm env}^p$ commute, then one can diagonalize $\gamma_{\rm env}$ via a symplectic and orthogonal transformation which changes neither the output entropy of the system nor the energy constraint. In the new basis the channel becomes a tensor product of monomodal Gaussian additive noise channels, for which it was proven that the Gaussian capacity is additive \cite{SEW05,H05}. Therefore, the optimal input and modulation covariance matrices are diagonal, as well as the noise covariance matrix. 

Then the symplectic eigenvalues are functions of the eigenvalues of the corresponding covariance matrices:
\begin{eqnarray}
  \overline{\nu}_i
  & =& \sqrt{\overline{\gamma}_i^q\,\overline{\gamma}_i^p} \, , \quad
  \overline{\gamma}_i^{q,p}
   = \gamma_{{\rm in},i}^{q,p} + \gamma_{{\rm mod},i}^{q,p} + \gamma_{{\rm env},i}^{q,p} \, , \nonumber\\
  \nu_{{\rm out},i}
  & = & \sqrt{\gamma_{{\rm out},i}^q  \gamma_{{\rm out},i}^p}\, , \quad
  \gamma_{{\rm out},i}^{q,p} =  \gamma_{{\rm in},i}^{q,p} + \gamma_{{\rm env},i}^{q,p} \, .
  \label{eq:nubar}
\end{eqnarray}
We denote the input energy allocated to channel $i$ as
\begin{equation}\label{eq:lambdai}
  \lambda_i \equiv \gamma^q_{{\rm in},i}+ \gamma^p_{{\rm in},i} +\gamma^q_{{\rm mod},i} + \gamma^p_{{\rm mod},i}.
\end{equation}
The energy (or photon number) constraint \eqref{eq:enconstr}, \eqref{eq:onelambda}, for $n$ modes can now be written as
\begin{equation}\label{eq:lambda}
  \lambda \equiv \sum\limits_{i=1}^n \lambda_i.
\end{equation}
The total output energy $\overline{\lambda}$ of the $n$-mode channel is the sum of the input energy $\lambda$ and the total energy of the noise $\lambda_{\rm env}$ (environment) 
\begin{equation}
  \overline{\lambda}=\lambda+\lambda_{\rm env}\, ,  \qquad
  \lambda_{\rm env} = \sum\limits_{i=1}^n\left(\gamma_{{\rm env},i}^q+\gamma_{{\rm env},i}^p\right).
\end{equation}

In \cite{SDKC09} the equipartition of the total output energy between the modes was obtained as the solution for such a model which we called \emph{global water-filling}, similar to the classical water-filling solution of $n$ parallel classical Gaussian channels \cite{Cover}. As in the one-mode case this solution only holds above a certain input energy threshold. Now we extend the discussion to energies below this threshold.

The maximum of the Holevo quantity of $n$ parallel channels is again determined using the Lagrange multipliers method (as in the one-mode case), where we now have $n$ purity constraints,
\begin{equation}\label{eq:pur}
  \gamma_{{\rm in},i}^{q}\gamma_{{\rm in},i}^{p} = \frac{1}{4},
\end{equation}
and the overall input energy constraint, \eqref{eq:lambda}. The Lagrangian is then constructed by a sum of $n$ Holevo $\chi$-quantities \eqref{eq:chi} for corresponding modes, $n$ Lagrange multipliers $\tau_i$ for the purity constraints, and only one common multiplier $\mu$ for the input energy constraint. The fact that the solution of the system of equations which results from the Lagrange multipliers method maximizes the Holevo quantity under the given constraints follows from the results on convex separable minimization subject to bounded variables found in \cite{S00}. This was first pointed out for a lossy channel in \cite{PLM09}. In connection to our problem it follows that the maximum is attained by our solution in the multimode case provided that, for the one-mode case, $\chi$ is a concave function of $\lambda$ on the solution. The proof of the concavity of $\chi$ in $\lambda$ on the one-mode solution is presented in Appendix \ref{sec:concavity}.

In general, the maximum of the Lagrangian corresponds to a partition of $n$ modes into three different sets, corresponding to one of three types of input energy distributions within each mode: the case of a quantum water-filling solution with four positive eigenvalues (see Appendix \ref{sec:qwf}), the case of one vanishing modulation eigenvalue with three positive eigenvalues (see Appendix \ref{sec:noqwf}), or the case when both modulation eigenvalues vanish and the mode does not participate in information transmission (i.e., unmodulated vacuum is sent). We denote the corresponding sets by: ${\mathcal N}_3$, ${\mathcal N}_2$ and ${\mathcal N}_1$. We denote the number of modes in the sets as $n_1$, $n_2$ and $n_3$, with $n = n_1 + n_2 + n_3$. Furthermore, we denote the input energies per set by $\lambda^{(1)},\lambda^{(2)},\lambda^{(3)}$, where
\begin{equation}\label{eq:lambda123}
	\lambda^{(1)} = \sum\limits_{i \in {\mathcal N}_1}{\lambda_i}, \; \lambda^{(2)} = \sum\limits_{i \in {\mathcal N}_2}{\lambda_i}, \; \lambda^{(3)} = \sum\limits_{i \in {\mathcal N}_3}{\lambda_i},
\end{equation}
which sum up to the total input energy $\lambda$ \eqref{eq:lambda}.

%%%%%%%%%%%%%%%%%%%%%%%%%%%%%%%%%%%%%%%%%
\subsubsection{Set ${\mathcal N}_3$: Modes with water-filling solution}\label{sec:setN3} 
%%%%%%%%%%%%%%%%%%%%%%%%%%%%%%%%%%%%%%%%%
For all modes that belong to ${\mathcal N}_3$ the water-filling solution described in section \ref{sec:qwf} holds. This means that the input energy allocated to each mode cannot be lower than its corresponding energy threshold, i.e., 
\begin{equation}\label{eq:threshni}
	\lambda_i \geq \lambda_{i,{\rm thr}}, \quad i \in {\mathcal N_3}.
\end{equation}
where $\lambda_{i,{\rm thr}}$ reads as in \eqref{eq:onethr} (for all $i$). Then for all $i \in {\mathcal N_3}$ the energy equipartition \eqref{eq:oneqwf} holds. Moreover Eq.~\eqref{eq:oneqwfmu} guarantees  that $\overline{\nu}_{\rm wf}$ is the common water-filling level for all modes due to the common Lagrange multiplier $\mu$, which is a monotonous function of $\overline{\nu}_{\rm wf}$, which now reads
\begin{equation}\label{eq:wfl}
  \overline{\nu}_{\rm wf} \equiv \overline{\gamma}^{q,p}_i = \frac{\overline{\lambda}^{(3)}}{2n_3}, \quad i \in {\mathcal N}_3,
\end{equation}
where $\overline{\lambda}^{(3)}$ is the total energy at the output of the modes belonging to set ${\mathcal N}_3$. 

As the partition of the input energy between the modes is \emph{a priori} not known the distribution of the modes between the sets is also not defined. However, we can determine whether a particular mode belongs to set ${\mathcal N}_3$ using the Lagrange multiplier $\mu$. This is possible, because as mentioned before, $\mu$ is a monotonically decreasing function of the input energy $\lambda$ and for $\lambda_i \geq \lambda_{{\rm thr},i}$ we have $\mu \leq \mu_{{\rm thr},i}$ defined in \eqref{eq:muthr} and depends only on the noise eigenvalues of mode $i$. Then we can formalize the definition of ${\mathcal N}_3$ using $\mu_{\mathrm{thr},i}$ as
\begin{equation}
	{\mathcal N}_3 = \{i | \; \mu_{\mathrm{thr},i} \geq \mu \}.
\end{equation}

If $\mu_{\mathrm{thr},i} \geq \mu$ for all $i$ then the set ${\mathcal N}_3$ contains all modes and we have
\begin{equation}
 \lambda^{(3)} = \lambda.
\end{equation}
In this case, Eqs.~\eqref{eq:oneqwf}, \eqref{eq:oneoptin}, and \eqref{eq:wfl} determine the global water-filling solution with
\begin{equation}\label{eq:nuglobalwf}
	{\overline{\nu}}_{\rm wf} = \frac{\overline{\lambda}}{2n}
\end{equation}
and
\begin{equation}\label{eq:nuoutglobalwf}
	\nu_{{\rm out},i} = \sqrt{\gamma_{{\rm env},i}^q\gamma_{{\rm env},i}^p} + \frac{1}{2}.
\end{equation}
If the condition \eqref{eq:threshni} is not satisfied for at least one mode, then this \emph{global water-filling} solution has no physical meaning because it will lead to negative values of some modulation eigenvalues. 
%%%%%%%%%%%%%%%%%%%%%%%%%%%%%%%%%%%%%%%%%%%%%%%%%%%%%%%%%%%%%%
\subsubsection{Set ${\mathcal N}_1$: Modes excluded from information transmission}\label{sec:setN1}
%%%%%%%%%%%%%%%%%%%%%%%%%%%%%%%%%%%%%%%%%%%%%%%%%%%%%%%%%%%%%%
Modes for which both modulation eigenvalues are 0 do not contribute to the Holevo quantity or, consequently to the information transmission. The zero modulation eigenvalues 
$\gamma_{{\rm mod},i}^{q,p} = 0, \, i \in {\mathcal N}_1$ imply
\begin{equation}\label{eq:gtgo}
  \overline{\gamma}^{q,p}_i=\gamma_{{\rm out},i}^{q,p}, \quad i \in {\mathcal N}_1.
\end{equation}
Obviously if the mode is not modulated there is no reason to spend input energy on the squeezing of this mode, which results in the vacuum state being the optimal input state
\begin{equation}\label{eq:purc}
 \gamma^{q,p}_{{\rm in},i}=\frac{1}{2}, \quad i \in {\mathcal N}_1.
\end{equation}
This is consistent with \eqref{eq:onebt} from which we obtain Eq.~\eqref{eq:purc} for $\lambda_i = 1$. 

In order to deduce the set of modes that are excluded from information transmission we can use the threshold value $\mu_{0,i}$ defined in Eq.~\eqref{eq:mu0} which corresponds to $\lambda_i = 1$ (vacuum energy) 
\[
	{\mathcal N}_1 = \{i | \; \mu \geq \mu_{0,i} \}.
\]
%%%%%%%%%%%%%%%%%%%%%%%%%%%%%%%%%%%%%%%%%%%%%%%%%%%%%%%%%%%%%%%%%%%%%%
\subsubsection{Set ${\mathcal N}_2$: Single-quadrature modulated modes}\label{sec:setN2}
For the modes for which $1 < \lambda < \lambda_{{\rm thr},i}$ the water-filling solution no longer holds. We have to set the modulation eigenvalue of the noisier quadrature to 0 in the same way as in the one-mode case. Again, as in the one mode case we assume that for each mode $i$ the $q$ quadrature is noisier than the $p$ quadrature. We can do this without loss of generality because, first, for a one-mode channel a swap of the noise quadratures does not change the one-shot capacity, and second, the one-shot capacity of the discussed multimode channel is additive. Then we have to set the modulation of the $q$ quadrature for all modes belonging to set ${\mathcal N}_2$ to 0, i.e.,
\begin{equation}\label{eq:mod0}
  \gamma_{{\rm mod},i}^{q} = 0, \quad i \in {\mathcal N}_2.
\end{equation}
This implies
\begin{equation}\label{eq:gn2}
  \overline{\gamma}^{q}_i=\gamma_{{\rm out},i}^{q}, \quad i \in {\mathcal N}_2.
\end{equation}
With the functions $\mu_{{\rm thr},i}$ and $\mu_{0,i}$ defined in \eqref{eq:muthr}, \eqref{eq:mu0} we can simply define this set, i.e.,
\begin{equation}
	{\mathcal N}_2 = \{i | \; \mu_{\mathrm{thr},i} < \mu < \mu_{0,i} \}.
\end{equation}

The eigenvalues that solve the optimization problem for the modes of this set are found using Eq.~\eqref{eq:onebt}, \eqref{eq:onebtmu}.

We note that both, $\mu_{0,i}$ and $\mu_{{\rm thr},i}$ depend only on the noise eigenvalues. Therefore, the partition into the three sets is completely determined by only one parameter $\mu$. Furthermore, we recall that $\mu$ is the common parameter which enters the equations for sets ${\mathcal{N}_2},{\mathcal{N}_3}$.
%%%%%%%%%%%%%%%%%%%%%%%%%%%%%%%%%%%%%%%%%%%%%%%%%%
\section{Solution for arbitrary number of modes}\label{sec:solarb}
%%%%%%%%%%%%%%%%%%%%%%%%%%%%%%%%%%%%%%%%%%%%%%%%%%
\subsection{Finite number of modes}
%%%%%%%%%%%%%%%%%%%%%%%%%%%%%%%%%%%%%%%%%%%%%%%%%%
Recall that the solution of the problem for $n$ modes is given by the optimal distribution of the input energy between the modes.
The optimal energy distribution within one mode depends on its corresponding set, which is given by the noise spectrum and the global parameter $\mu$. 

Now we present the algorithm that allows us to find the solution of our optimization problem. First, we further develop the equations that correspond to modes of set ${\mathcal N}_2$. We call the right-hand side of Eq.~\eqref{eq:onebt}
\begin{equation}\label{eq:f}
	f(\gamma_{{\rm in},i}^{q}) \equiv \frac{g'\left(\nu_{{\rm out},i}-\frac{1}{2}\right)}{2 \nu_{{\rm out},i}}
      		\left(\gamma_{{\rm out},i}^{p}
  - \frac{\gamma_{{\rm in},i}^{p}}{\gamma_{{\rm in},i}^{q}}\gamma_{{\rm out},i}^{q}
          \right).
\end{equation}
We note that for given noise eigenvalues $f$ is a function of only one independent variable $\gamma^q_{{\rm in},i}$. Using definition \eqref{eq:lambdai} and the fact that, for the modes belonging to the second set $\gamma_{{\rm mod},i}^q = 0$, we rewrite
\begin{equation}\label{eq:diffgbarqp}
	\overline{\gamma}_i^{p}-\overline{\gamma}_i^{q} = \lambda_i - 2	\gamma_{{\rm in},i}^{q} - \gamma_{{\rm env},i}^{q} + \gamma_{{\rm env},i}^{p}.
\end{equation}
Furthermore, from \eqref{eq:onebtmu} we express
\begin{equation}\label{eq:fracmugbarq}
	\frac{g'(\overline{\nu}_i - \frac{1}{2})}{2 \overline{\nu}_i} = \frac{\mu}{{\overline{\gamma}}^q_i}.
\end{equation}
Then we insert Eqs.~\eqref{eq:diffgbarqp} and \eqref{eq:fracmugbarq} together into the left hand side of Eq.~\eqref{eq:onebt} and obtain with $\overline{\gamma}^{q}_i=\gamma_{{\rm out},i}^{q}$ and \eqref{eq:f}
\begin{equation}\label{eq:lambdaiset2}
	\lambda_i(\gamma_{{\rm in},i}^{q},\mu) = \frac{\gamma_{{\rm out},i}^{q}}{\mu} f(\gamma_{{\rm in},i}^{q}) + 2 \gamma_{{\rm in},i}^{q} + \gamma_{{\rm env},i}^{q} - \gamma_{{\rm env},i}^{p}.
\end{equation}
This means that we have established a relation between the optimal input eigenvalues $\gamma_{{\rm in},i}^{q}, \; i \in {\mathcal N}_2$, the global parameter $\mu$ and the optimal input energy distribution $\lambda_i$ between the modes in ${\mathcal N}_2$. Using Eq.~\eqref{eq:lambdaiset2} and definition \eqref{eq:lambdai} we can now eliminate variable $\lambda_i$ from Eq.~\eqref{eq:diffgbarqp}. Thus, we obtain a transcendental equation that determines the optimal input eigenvalues $\gamma_{{\rm in},i}^q$ as an implicit function of $\mu$:
\begin{equation}\label{eq:optin}
	g'\left(\gamma_{{\rm out},i}^{q}\sqrt{1 +  \frac{f(\gamma_{{\rm in},i}^{q})}{\mu}} - \frac{1}{2}\right)
	= 2\mu \, \sqrt{1 +  \frac{f(\gamma_{{\rm in},i}^{q})}{\mu}}.
\end{equation}

Now we are ready to calculate the input energies of all three sets for a given $\mu$. First, we evaluate the total input energy of ``water-filling'' modes, i.e. modes in ${\mathcal N}_3$. Using Eqs.~\eqref{eq:lambdai}, \eqref{eq:lambda123} and \eqref{eq:wfl} we deduce the total input energy used for the modes in ${\mathcal N}_3$ as a function of $\mu$
\begin{equation}\label{eq:lambda3mu}
	\lambda^{(3)}(\mu) = \sum\limits_{i \in {\mathcal N}_3}{\left(2\overline{\nu}_{\rm wf}(\mu)-\gamma_{{\rm env},i}^{q}-\gamma_{{\rm env},i}^{p}\right)}.
\end{equation}
Second, the total (vacuum) energy of modes belonging to ${\mathcal N}_1$, using \eqref{eq:mu0}, reads
\begin{equation}
	\lambda^{(1)}(\mu) = \sum\limits_{i \in {\mathcal N}_1}\, 1 = n_1.
\end{equation}
Functions $\lambda^{(1)}(\mu)$ and $\lambda^{(3)}(\mu)$ depend on $\mu$ through ${\mathcal N}_1, {\mathcal N}_3$ and $\overline{\nu}_{\rm wf}$, which are only functions of $\mu$ and the noise eigenvalues. The total input energy for modes in ${\mathcal N}_2$ is the sum
\begin{equation}\label{eq:lambda2mu}
	\lambda^{(2)}(\mu) = \sum\limits_{i \in {\mathcal N}_2}{\lambda_i(\gamma_{{\rm in},i}^{q},\mu)}.
\end{equation}
Now we apply the overall input energy constraint
\begin{equation}\label{eq:lambda123mu}
	\lambda^{(1)}(\mu) + \lambda^{(2)}(\mu) + \lambda^{(3)}(\mu) = \lambda.
\end{equation}
Thus, we obtain a closed equation on $\mu$ which we can solve by iterations. The solution of this system of equations provides us with $n_2$ optimal eigenvalues $\gamma_{{\rm in},i}^{q}$ and $\mu$ which determine all other eigenvalues. Once the optimal spectra are obtained one can calculate the one-shot capacity of $n$ modes ($n$ successive uses) of the channel using Eqs.~\eqref{eq:ncapacity} and \eqref{eq:S}
\begin{equation}\label{eq:oscapacitygfun}
	 C_1(T^{(n)}) = \sum_{i=1}^n\left( g\left(\overline{\nu}_i -\frac{1}{2}\right) - g\left(\nu_{{\rm out},i}-\frac{1}{2}\right) \right),
\end{equation}
where here $\overline{\nu}_i,\nu_{{\rm out},i}$ contain the obtained optimal input and modulation spectra. This one-shot capacity of the $n$-mode channel may be considered as $n$ times the capacity of the one mode channel.

%%%%%%%%%%%%%%%%%%%%%%%%%%%%%%%%%%%%%%%%%%%%%%%%%%
\subsection{Infinite number of modes}
%%%%%%%%%%%%%%%%%%%%%%%%%%%%%%%%%%%%%%%%%%%%%%%%%%
In order to make the transition to an infinite number of channel uses we have to consider a parallel channel with an infinite number of one mode channels, $n \rightarrow \infty$. In this limit all functions previously labeled with $i$ depend now on a continuous parameter $x$ defined on a proper domain which depends on the particular model. All sums that run from $i = 1,...,n$ now become integrals over the whole domain of $x$. The three sets become now sets of continuous variables and cover the whole domain of $x$; they read
\begin{equation}\label{eq:N1N2N3cont}
	\begin{split}
		{\mathcal N}_1 & = \{x | \; \mu_0(x) \le \mu\},\\
		{\mathcal N}_2 & = \{x | \; \mu_{\mathrm{thr}}(x) < \mu < \mu_{0}(x) \},\\
		{\mathcal N}_3 & = \{x | \; \mu_{\mathrm{thr}}(x) \geq \mu\},
	\end{split}
\end{equation}
where $\mu_{\mathrm{thr}}(x), \mu_{0}(x)$ are defined as in \eqref{eq:muthr}, \eqref{eq:mu0} where index $i$ is replaced by $x$. Equations \eqref{eq:optin}-\eqref{eq:lambda2mu} remain the same, except for the replacements $\gamma_{{\rm in},i}^q$ by $\gamma_{\rm in}^q(x)$ and the sums over $i$ by integrals over $x$.

Once the $\mu$ is found which is the solution of \eqref{eq:lambda123mu} we can determine the optimal spectra $\gamma_{\rm in}^{q,p}(x)$ and $\gamma_{\rm mod}^{q,p}(x)$. The found optimal spectra are used to evaluate the capacity 
\begin{equation}\label{eq:capacity}
	\begin{split}
	C & = \lim_{n \rightarrow \infty}{\frac{1}{n}C_1(T^{(n)})}\\
	  & = \frac{1}{|A|}\int\limits_{x \in A}{dx \, \left( g\left(\overline{\nu}(x) -\frac{1}{2}\right) - g\left(\nu_{\rm out}(x)-\frac{1}{2}\right) \right)},
	\end{split}
\end{equation}
where $A$ is the spectral domain of $x$ and $|A|$ is its size. In the case of a global water-filling, i.e. if $\mu_{\rm thr}(x) \geq \mu \; \forall x$ then \eqref{eq:capacity} simplifies with \eqref{eq:nuglobalwf} and \eqref{eq:nuoutglobalwf} to
\begin{equation}\label{eq:capacityglobalwf}
	\begin{split}
	C & = g\left(\overline{n} + \frac{1}{2|A|}\int_{x \in A}{dx \, \{\gamma_{\rm env}^q(x) + \gamma_{\rm env}^p(x) \}}\right)\\
	  & - \frac{1}{|A|}\int_{x \in A}{dx  \, g\left(\sqrt{\gamma_{\rm env}^q(x)\gamma_{\rm env}^p(x)}\right)},
	\end{split}
\end{equation}
where $\gamma_{\rm env}^{q,p}(x)$ are the noise eigenvalue spectra.
\begin{figure}
	\centering
		\includegraphics[width=0.5\textwidth]{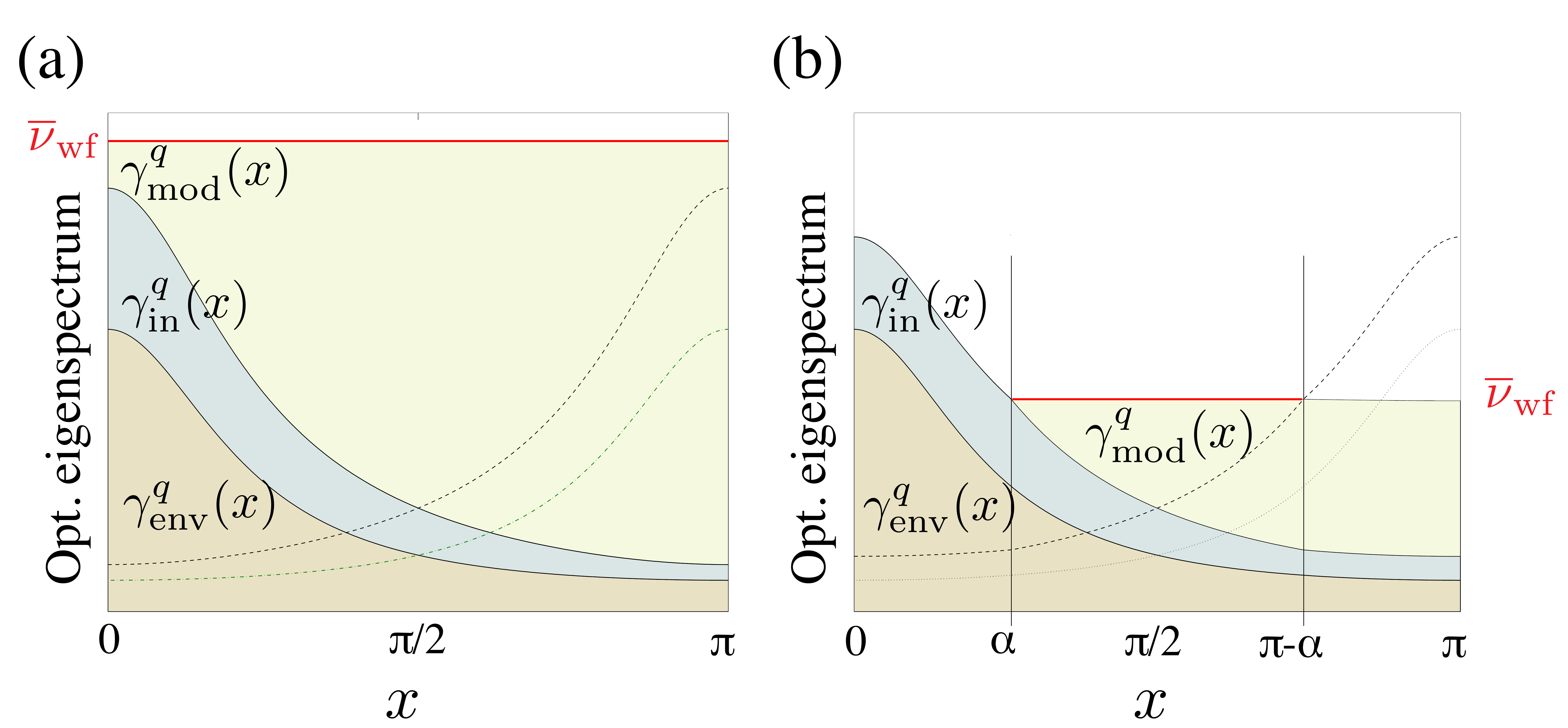} 		 
	\caption{(Color online) Stacked area plot: optimal input and modulation eigenvalue spectra $\gamma^q_{\rm in}(x),\gamma^q_{\rm mod}(x)$ and noise spectrum $\gamma^q_{\rm env}(x)$ (for $p$ in dashed and dotted curves) for a particular $\phi$ and $N$. (a) Global quantum water-filling solution, where $\lambda > \lambda_{\rm thr}$. $\overline{\nu}_{\rm wf}$ (solid bar) denotes the water-filling level. (b) Below threshold: Modes with $x \in [\alpha,\pi-\alpha]$ belong to ${\mathcal N}_3$ with water-filling level $\overline{\nu}_{\rm wf}$ (solid bar). Modes with $x \in [0,\alpha]$ and $x \in [\pi-\alpha,\pi]$ belong to set ${\mathcal N}_2$, where the modulation is below $\overline{\nu}_{\rm wf}$.}
	\label{fig:markqwf}
\end{figure}
\begin{figure}
	\centering
		\includegraphics[width=0.5\textwidth]{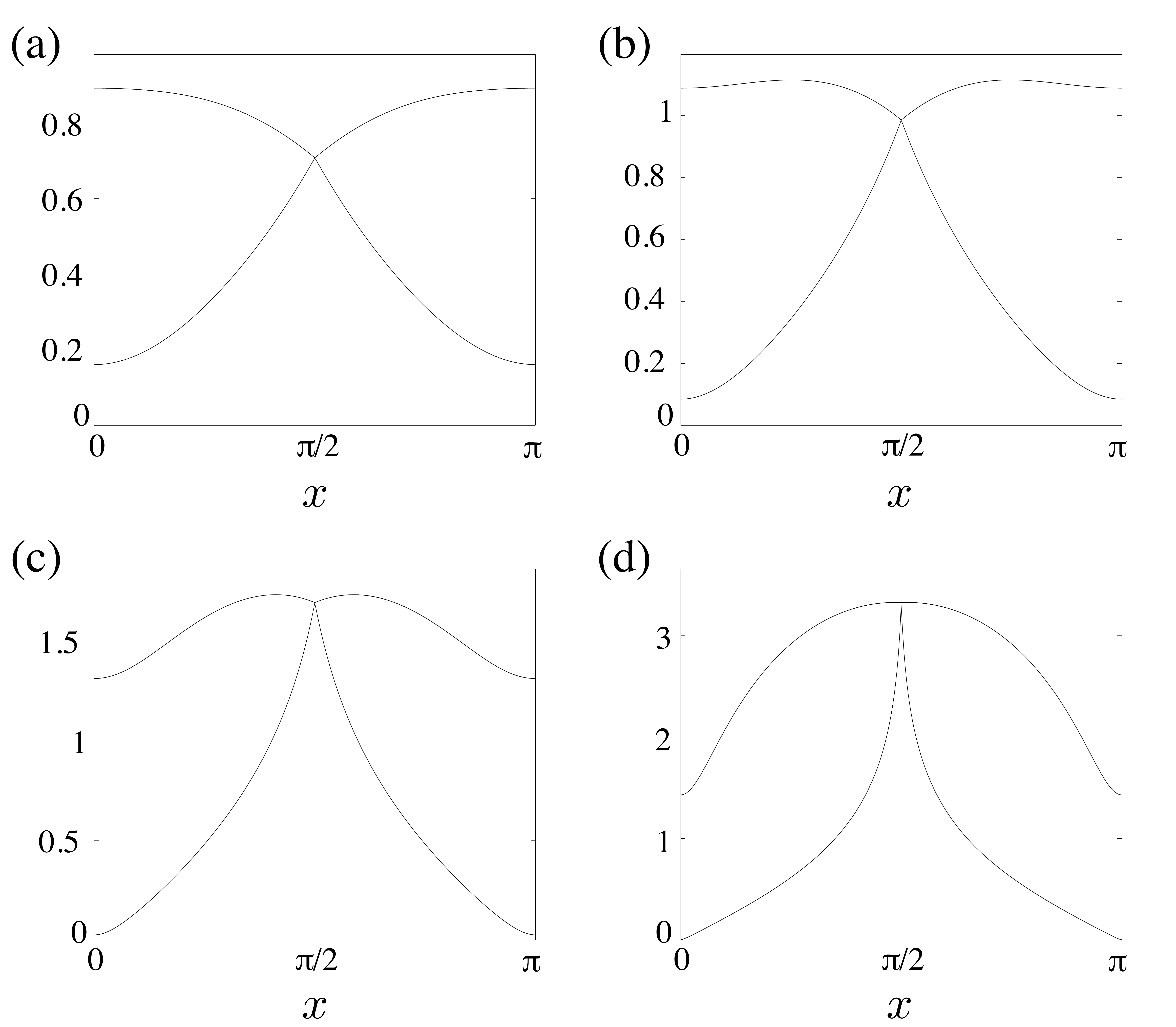} 
	\caption{Functions $\mu_0(x)$ (upper curve) $ \mu_{\rm thr}(x)$ (lower curve) for (a) $\phi = 0.5$, (b) $\phi = 0.7$, (c) $\phi = 0.9$, (d) $\phi = 0.99$. For all plots we took $N=1$.}
	\label{fig:mufuns}
\end{figure}
	
%%%%%%%%%%%%%%%%%%%%%%%%%%%%%%%%%%%%%%%%%%%%%%%%%%
\section{Gauss-Markov channel}\label{sec:gaussmarkov}
%%%%%%%%%%%%%%%%%%%%%%%%%%%%%%%%%%%%%%%%%%%%%%%%%%
Now we consider a particular channel where the correlations of the noise are modeled by a Gauss-Markov process. Note, that for the whole class of noises with correlations given by stationary (shift-invariant) Gauss processes, the covariance matrix that describes such noise is a symmetric Toeplitz matrix. A straightforward example is the Markov process of order ${\mathcal P}$, also called autoregressive process with white Gaussian noise \cite{SZ99} (see Appendix \ref{sec:arp}). Our above treatment applies to such channels. However, in the following, we restrict to the simplest case of ${\mathcal P}=1$.
We have already introduced the covariance matrix of such a noise in \cite{SDKC09} and state here only its definition. Essentially this is a classical noise with nearest-neighbor correlations in the $q$ quadratures (anti-correlations in the $p$ quadratures). Its covariance matrix contains Toeplitz matrices for the $q$ quadrature and $p$ quadrature, so that
\begin{equation}
  \gamma_\mathrm{env} = 
  \begin{pmatrix}
  M(\phi) & 0\\
	0 & M(-\phi)
	\end{pmatrix},
	\label{eq:marknoise}
\end{equation} 
with
\begin{equation}
	M_{ij}(\phi) = N\phi^{|i-j|}, \quad 0 \leq \phi < 1, \, N \ge 0,
   \label{eq:markmat}
\end{equation}
where $\phi$ is the correlation parameter and $N$ is the variance of the noise. We remark that a Toeplitz matrix has the same values on each $k=|i-j|$th diagonal.

Matrix \eqref{eq:marknoise} can be diagonalized in the limit of\\ $n \rightarrow \infty$ using a passive symplectic transformation \cite{SDKC09}, which allows us to study the channel system in the diagonal, noncorrelated basis, since entropies and the input energy constraint remain unchanged by such transformations.

The spectra of the noise quadratures in the limit of an infinite number of uses of \eqref{eq:marknoise} read \cite{SDKC09}
\begin{equation}\label{eq:genvqp}
	\gamma_{\rm env}^{q,p}(x) = N \, \frac{1 - \phi^2}{1 + \phi^2 \mp 2\phi \cos(x)}, \quad x \in [0,2\pi],
\end{equation}
with the upper sign for the $q$ quadrature and the lower sign for the $p$ quadrature. As this spectrum is mirror symmetric with respect to $x=\pi$ and since the Gaussian capacity of the channel is additive, we reduce the spectral domain to $x \in [0,\pi]$.

In order to find the capacity of the channel for given noise parameters $N,\phi$, we first check whether the threshold condition \eqref{eq:threshni} (here with continuous parameter $x$ replacing index $i$) is fulfilled for all $x$. If this is the case, then the solution is a global water-filling, depicted in Fig.~\ref{fig:markqwf} (a), where optimal eigenvalue spectra are obtained by \eqref{eq:oneoptin}, \eqref{eq:oneqwf} (for all $x$), and the capacity \eqref{eq:capacityglobalwf} can be simplified to \cite{SDKC09}
\begin{equation}\label{eq:Cat}
	\begin{split}
	C & = g({\overline{n}}+N) - \frac{1}{\pi}\int_0^\pi{d x \, g\left(\sqrt{\gamma_{\rm env}^{q}(x)\gamma_{\rm env}^{p}(x)}\right)},\\ 	
	\overline{n} & \geq \frac{2\phi}{1-\phi}\left(N+\frac{1}{2}\right) \equiv \frac{\lambda_{\rm thr}-1}{2}.
	\end{split}
\end{equation}
If the threshold condition is violated then we can apply the algorithm which is presented in section \ref{sec:solarb}. 

We note that for different noise parameters the threshold functions $\mu_0(x)$, $\mu_{\rm thr}(x)$ may have a complicated profile as depicted in Fig.~\ref{fig:mufuns} for different noise parameter values. In Fig. \ref{fig:mufixednoise} we illustrate, for a particular choice of the noise parameters ($N$, $\phi$) different partitions of the spectral domain between the sets for different input energies $\lambda$ corresponding to different $\mu$. 

Our result confirms that the modes belonging to ${\mathcal N}_2$ are squeezed in the less noisy quadrature which is the one that is modulated, as depicted in Fig.~\ref{fig:optinmod}. An example plot of optimal input and modulation spectra for $\lambda < \lambda_{\rm thr}$ is shown in Fig.~\ref{fig:markqwf} (b). 
\begin{figure}
		\includegraphics[width=0.5\textwidth]{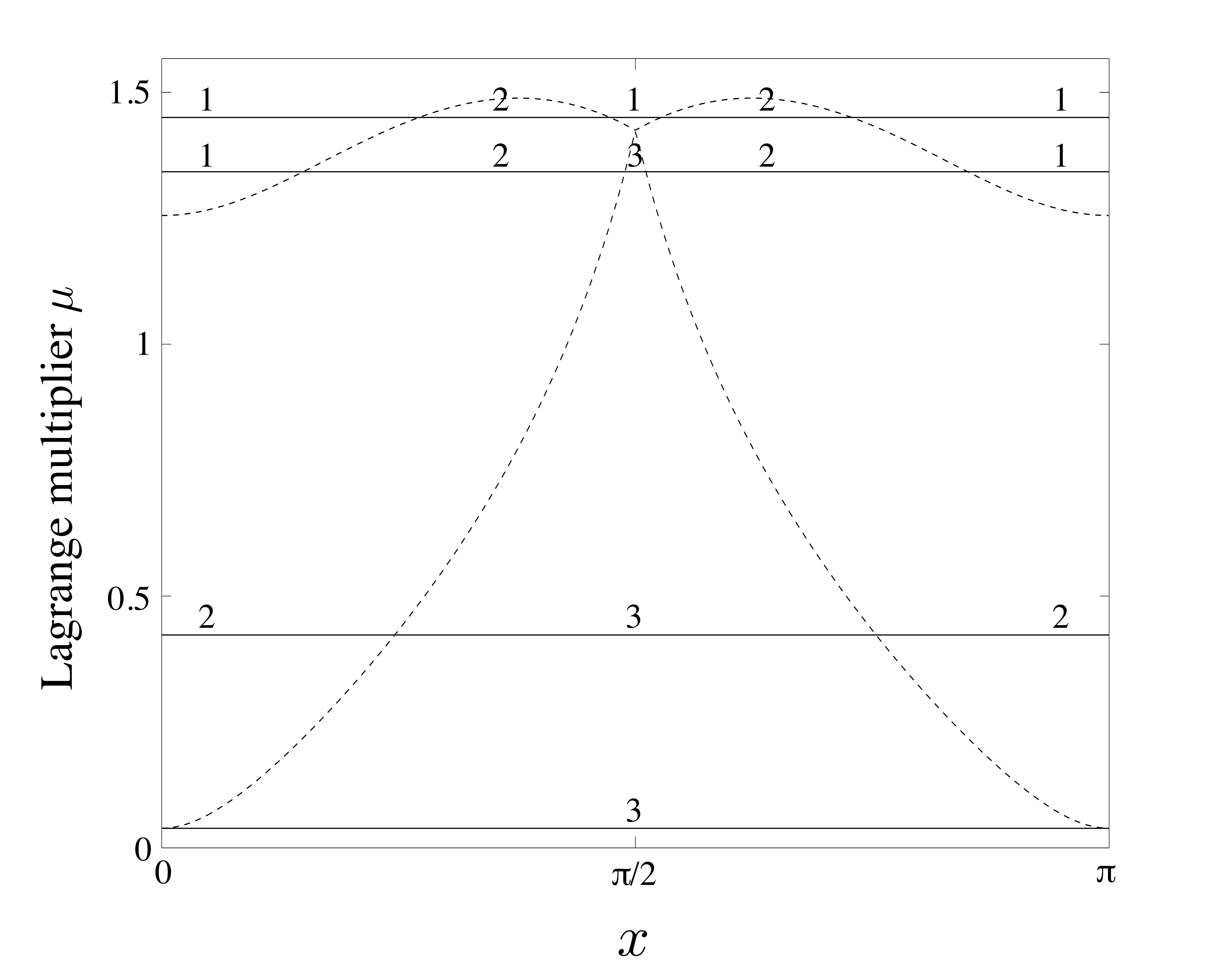}
	\caption{Functions $\mu_{\rm thr}(x)$ (lower dashed curve), $\mu_0(x)$ (upper dashed curve) and values for $\mu$ (solid bars) for different input energies $\lambda$, and noise parameters $\phi=0.85, N=1$. From top to bottom the values are $\mu=1.45 (\lambda = 1.006), 1.34 (\lambda = 1.04), 0.42 (\lambda = 3), 0.04 (\lambda = 35)$. The numbers indicate the intervals on the $x$-axis that belong to sets ${\mathcal N}_1,{\mathcal N}_2$ or ${\mathcal N}_3$.}
	\label{fig:mufixednoise}
\end{figure}
We see the naturally expected behavior of the capacity in Fig.~\ref{fig:capacityplot}. It decreases with increasing noise variance $N$ and increases with increasing noise correlations $\phi$. We note that the capacity increases with $\phi$ up to the noiseless capacity at ``full correlations'' ($\phi \rightarrow 1$). This limit will be discussed in section~\ref{sec:fullcorr}.
\begin{figure}
	\centering
		\includegraphics[width=0.5\textwidth]{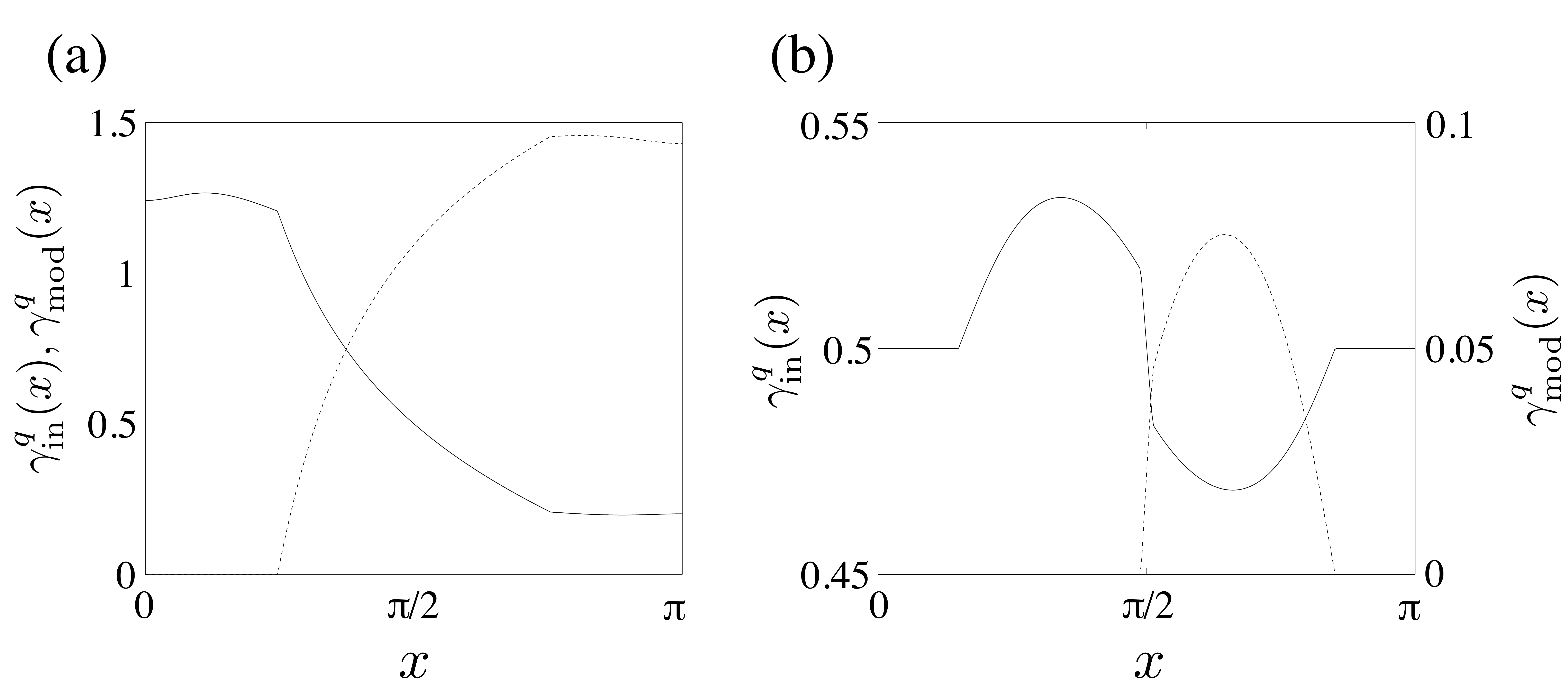}
	\caption{Optimal input $\gamma_{\rm in}^q(x)$ (solid curve) and modulation $\gamma_{\rm mod}^q(x)$ (dashed curve) eigenvalue spectra of the $q$ quadrature ($p$ quadrature spectra are the same but mirrored with respect to a vertical line at $\pi/2$) vs. spectral parameter $x$, for $\phi = 0.85, N = 1$ and $\lambda < \lambda_{\rm thr}$. The partitioning in sets is taken from Fig.~\ref{fig:mufixednoise}. In (a): $\lambda = 3$ which corresponds to $\mu = 0.42$, in (b): $\lambda = 1.04$ which corresponds to $\mu = 1.34$.}
	\label{fig:optinmod}
\end{figure}
\begin{figure}
	\centering		
		\includegraphics[width=0.5\textwidth]{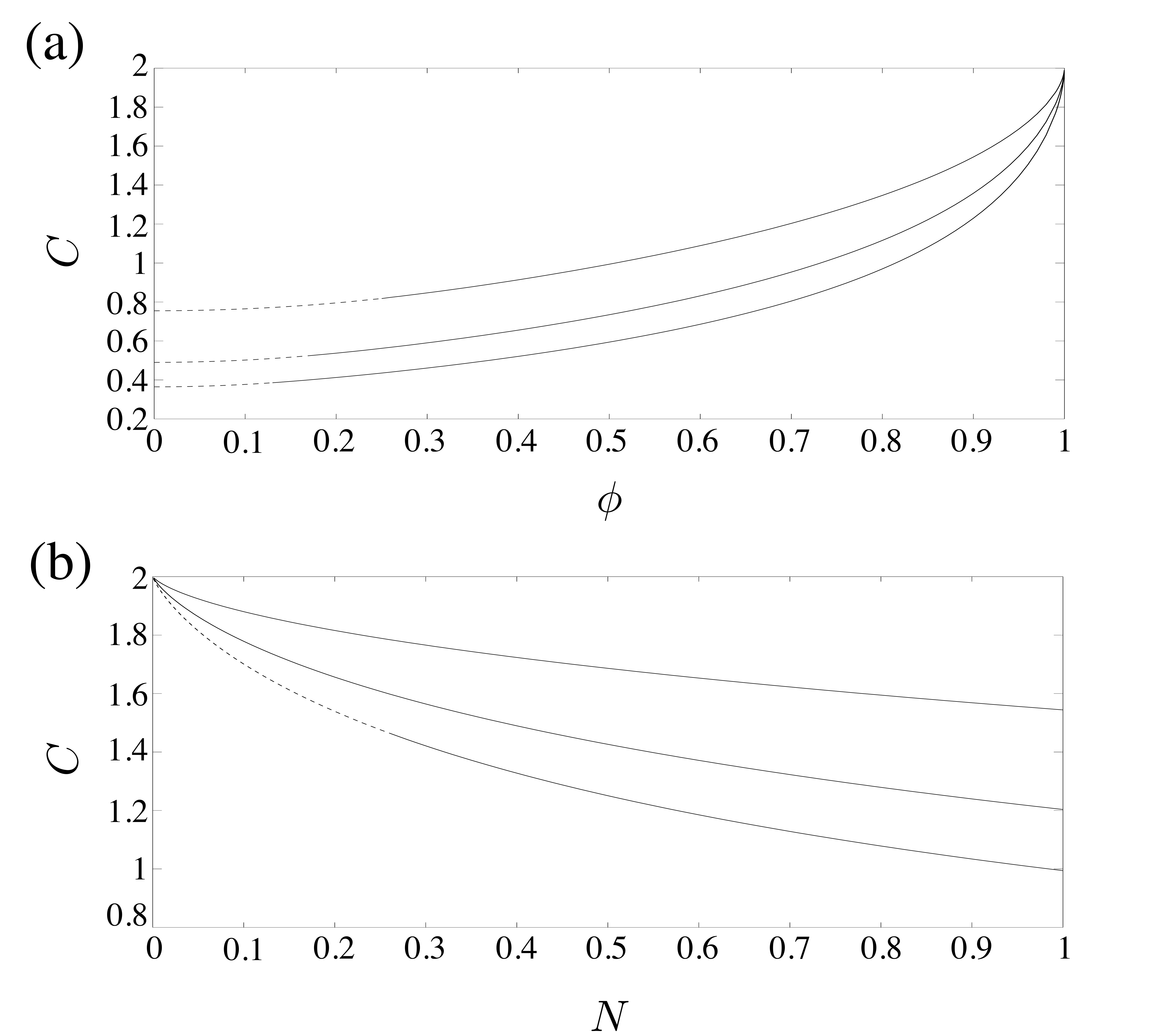}
	\caption{(a) Capacity $C$ (in bits) vs. correlation $\phi$, where from top to bottom $N = 1,2$ and $3$. (b) Capacity $C$ (in bits) vs. noise variance $N$, where from top to bottom $\phi = 0.9,0.7$ and $0.5$. The input energy is $\lambda = 3$ for both plots. The dashed part of the curves corresponds to the global water-filling solution with $\lambda > \lambda_{\rm thr}$. One observes that the capacity for full correlations $\phi \rightarrow 1$ tends to the capacity of the ideal noiseless channel $N=0$.}
	\label{fig:capacityplot}
\end{figure}

%%%%%%%%%%%%%%%%%%%%%%%%%%%%%%%%%%%%%%%%%%%%%%%%%%
\subsection{Optimal quantum input state}
%%%%%%%%%%%%%%%%%%%%%%%%%%%%%%%%%%%%%%%%%%%%%%%%%%
An important issue is to derive the covariance matrix of the optimal input state in the original ``correlated'' basis. We know that in the basis where the noise, modulation, and input matrices are diagonal, the optimal input spectrum corresponds to a product of one-mode squeezed states. By using general properties of Toepltiz matrices we conclude in Appendix \ref{sec:toep} that in the limit of an infinite number of modes the optimal input covariance matrix in the original basis is also Toeplitz and we found that its $k$th diagonal reads
\begin{equation}\label{eq:ginqpk}
	\gamma_{{\rm in},k}^{q,p} = \frac{1}{2 \pi}{\int\limits_0^{2\pi}{d x \, e^{i k x}} \, \gamma_{\rm in}^{q,p}(x)} , \quad k = 0,1,2,...,\infty.
\end{equation}
We can express $\gamma_{\rm in}^{q,p}(x)$ exactly in the case of global water-filling ($\lambda \geq \lambda_{\rm thr}$) which we consider for the rest of this subsection, that is,
\[
	\gamma_{\rm in}^{q,p}(x) = \frac{1}{2}\sqrt{\frac{\gamma_{\rm env}^{q,p}(x)}{\gamma_{\rm env}^{p,q}(x)}}.
\]
Inserting the definition of the noise spectrum of the Gauss-Markov channel \eqref{eq:genvqp} we deduce the spectra for the $q$ and $p$ quadrature of the optimal input matrix, i.e.,
\begin{equation}\label{eq:gin}
	\gamma^{q,p}_{{\rm in},k} = \frac{1}{4\pi}\int\limits_0^{2 \pi}{d x \, e^{i k x} \, \sqrt{\frac{1 + \phi^2 \pm 2 \phi \cos(x)}{1 + \phi^2 \mp 2 \phi \cos(x)}}},
\end{equation}
where the upper sign is for $q$ and the lower for $p$. In order to verify that the overall state is entangled we can check whether the reduced single mode states are mixed, i.e. whether for the reduced covariance matrix we have
\begin{equation}\label{eq:gin0}
	\det{\gamma_{\rm in}} = \gamma^{q}_{{\rm in},0}\gamma^{p}_{{\rm in},0} > \frac{1}{4}, \quad \phi > 0.
\end{equation}
Integration over the whole domain $0$ to $2\pi$ leads to $\gamma^{q}_{{\rm in},0} = \gamma^{p}_{{\rm in},0}$. Then we find for $\gamma^q_{{\rm in},0}(\phi = 0)$ = 1/2, which means that in the absence of correlations the optimal input state is a set of coherent states and not entangled. The limit of $\phi \rightarrow 1$ is unphysical because in this limit each single mode state becomes a thermal state with its temperature tending to infinity. This corresponds to an overall maximally entangled state. It is easy to show that \eqref{eq:gin0} is monotonically increasing from $\phi = 0$ to $\phi = 1$ and therefore we conclude that for all $\phi > 0$ the optimal input state is entangled.

In order to express the covariance matrix of the overall modulated output, let us recall that in the global water-filling case the overall modulated output eigenvalues are identical [${\overline{\gamma}}^{q,p}(x)$ is constant in $x$] and we can express them using Eq.~\eqref{eq:nuglobalwf} as ${\overline{\gamma}}^{q,p}(x) = {\overline{n}} + N + 1/2$.
Therefore, from Eq.~\eqref{eq:ginqpk} we easily see that only the main diagonal ($k=0$) has nonvanishing values and these values are identical. Then the covariance matrix $\overline{\gamma}$ is proportional to the identity matrix $I$, and therefore it is diagonal in the initial as well as in the rotated basis is ${\overline{\gamma}}^{q,p} = I({\overline{n}} + N + 1/2)$. This means that the sum of the optimal input and modulation covariance matrix has to cancel the correlations of the noise. 

%%%%%%%%%%%%%%%%%%%%%%%%%%%%%%%%%%%%%%%%%%%%%%%%%%
\subsection{Full correlations}\label{sec:fullcorr}
%%%%%%%%%%%%%%%%%%%%%%%%%%%%%%%%%%%%%%%%%%%%%%%%%%
We observe in Fig.~\ref{fig:capacityplot} that for fixed $N$ and $\lambda$ the higher the correlations are, the higher is the capacity. Furthermore, for $\phi \rightarrow 1$ the capacity tends to the capacity of the noiseless channel 
\begin{equation}\label{eq:Clim}
	\lim_{\phi \rightarrow 1} C = C_0 = g(\overline{n}),
\end{equation}
where the noiseless capacity $C_0$ was calculated in e.g. \cite{GGLMSY04}. Equation \eqref{eq:Clim} can be deduced by the following reasoning. For any $0 < \phi < 1$ the capacity $C$ is upper bounded by $C_0$. In addition, $C$ is lower bounded by the optimal transmission rate when using coherent states, which is denoted by $R$ in the following. Thus, we need to show that for $\phi \rightarrow 1$, both bounds fall together.

$R$ is easily calculated, as the restriction to coherent input states basically leads to the discussion of two classical multimode Gaussian channels with noises $\tilde{\gamma}^{q,p}_{\rm env}(x) = \gamma^{q,p}_{\rm env}(x) + 1/2$. Clearly, the solution of the optimization problem completely reduces now to the classical water-filling \cite{Cover} which determines the optimal modulation spectrum. For the noise spectrum of the Gauss-Markov channel \eqref{eq:genvqp} we find the global water-filling solution 
\begin{equation}\label{eq:Rat}
	\begin{split}
		R & = g(\overline{n} + N) -\frac{1}{\pi}\int\limits_0^\pi{d x \, g\left( \sqrt{\tilde{\gamma}^{q}_{\rm env}(x) \tilde{\gamma}^{p}_{\rm env}(x)} - \frac{1}{2}\right)},\\
		\overline{n} & \geq \frac{2\phi}{1-\phi} N.
	\end{split}
\end{equation}
Since for global water-filling the overall modulated output state is identical in $R$ and $C$ \eqref{eq:Cat} the difference to the capacity comes from the difference in the nonmodulated output only and is remarkably little. Indeed if we look at the output eigenvalue spectrum for $R$ in \eqref{eq:Rat} which reads 
\[
	\tilde{\nu}_{\rm out}(x) = \sqrt{\frac{1}{4} + \frac{\gamma_{\rm env}^q(x) + \gamma_{\rm env}^p(x)}{2} + \gamma_{\rm env}^q(x)\gamma_{\rm env}^p(x)}
\]
instead of 
\[
	\nu_{\rm out}(x) = \sqrt{\frac{1}{4} + \sqrt{\gamma_{\rm env}^q(x)\gamma_{\rm env}^p(x)} + \gamma_{\rm env}^q(x)\gamma_{\rm env}^p(x)}
\]
for $C$ \eqref{eq:Cat} we see that the two formulas simply differ by the terms which are the arithmetic mean of the noise eigenvalues $\gamma_{\rm env}^q(x)$, $\gamma_{\rm env}^p(x)$ in the first one and the geometric mean in the second one. As the geometric mean is always less or equal than the arithmetic mean one confirms that $C \geq R$.

Below the energy threshold one has to solve
\begin{equation}
	\frac{\pi - \alpha}{\pi}\gamma_{\rm env}^q(\alpha) = \frac{1}{\pi}\int\limits_\alpha^\pi{d x \gamma_{\rm env}^q(x)} + \overline{n}
\end{equation}
for $\alpha$ (depicted in Fig.~\ref{fig:markqwf} (b) but here with $\gamma_{\rm in}^{q,p}(x) = 1/2$) which is the $x$ value that defines the sets ${\mathcal N}_2$ and ${\mathcal N}_3$ when restricted to coherent states. For the found value of $\alpha$ we determine the optimal modulation eigenvalues 
\begin{equation}\label{eq:Roptmod}
	\begin{split}
		\tilde{\gamma}^q_{\rm mod}(x) & = \theta(x-\alpha)[\gamma^q_{\rm env}(\alpha) - \gamma^q_{\rm env}(x)],\\
		\tilde{\gamma}^p_{\rm mod}(x) & = \theta(\pi-\alpha-x)[\gamma^q_{\rm env}(\alpha) - \gamma^p_{\rm env}(x)], 		
	\end{split}
\end{equation}
where $\theta(x)$ is the Heaviside step function. By inserting $\tilde{\gamma}^{q,p}_{\rm in}(x) = 1/2, \tilde{\gamma}^q_{\rm mod}(x)$ and $\tilde{\gamma}^{q,p}_{\rm env}(x)$ in \eqref{eq:capacity} one obtains $R$ for $\overline{n} < 2N\phi/(1-\phi)$.

In the limit of full correlations $\phi \rightarrow 1$ the noise spectra $\gamma_{\rm env}^{q,p}(x)$ tend to 0 for $0 < x < \pi$ and to infinity for $x = 0$ (for the $q$-spectrum) and $x = \pi$ (for the $p$-spectrum). Due to the finite energy of the noise, $\frac{1}{\pi}\int_0^\pi{d x \gamma_{\rm env}^{q,p}(x) = N}$, these functions become delta-like distributions. In this limit $\alpha \rightarrow 0$ and the solution to $R$ is given by a classical water-filling over a vacuum noise spectrum with infinite edges which though can be shown to give an infinitesimally small contribution and therefore can be neglected. Thus \eqref{eq:Clim} is proven. The same result was obtained in \cite{LMM09} for a channel with additive Markov noise, which becomes a collection of thermal channels when the noise is diagonalized.
%%%%%%%%%%%%%%%%%%%%%%%%%%%%%%%%%%%%%%%%%%%%%%%%%%
\subsection{How useful are the optimal input states?}
%%%%%%%%%%%%%%%%%%%%%%%%%%%%%%%%%%%%%%%%%%%%%%%%%%
In this subsection we evaluate the gain $G$ from the use of the optimal input states with respect to coherent product states for the Gauss-Markov channel for two modes and an infinite number of modes. Our motivation here is that the optimal input states are entangled and therefore may be not easy to generate. On the contrary, coherent states are easily accessible in the laboratory by standard tools of quantum optics. The gain $G$ is given by the ratio of the capacity $C$ to the optimal transmission rate using coherent states $R$ (discussed in section \ref{sec:fullcorr})
\begin{equation}
	G \equiv \frac{C}{R}.
\end{equation}
The gain was already discussed in \cite{CCMR05} for the case of a two mode additive channel which is identical to our Gauss-Markov channel with noise covariance matrix \eqref{eq:marknoise} taking $n = 2$. We remark here, that the capacity of two modes with correlations $\phi$ (in $q$ and $p$, but no $q-p$ correlations) is identical to the capacity of the mono-modal phase dependent channel discussed in section \ref{sec:onemode}, because in the two-mode case the diagonalized noise spectrum leads to two phase-dependent monomodal channels with inverse variance in $q$ and $p$. As it was shown \cite{CCMR05} the gain for such a channel exhibits a maximum with respect to the input energy constraint $\overline{n}$ for fixed signal-to-noise ratio 
\[
	\mathrm{SNR} = \frac{\overline{n}}{N}
\]
and correlation $\phi$. Furthermore, it was deduced that the gain increases with increasing correlations $\phi$ between the two modes. 

In the case of an infinite number of modes, we know already that in the absence of correlations the optimal input states are coherent states, and therefore there is no gain ($G=1$). For full correlations, the behavior is essentially different from the two-mode case: since the channel becomes effectively noiseless, coherent input states are optimal in this limit as well, whereas for two modes the highest available squeezing is best. Therefore, an interesting question is where we find the maximum gain with respect to the noise correlations in the limit of an infinite number of uses of the channel. In Fig.~\ref{fig:markGainCompvsnmax} 
\begin{figure}
	\centering
		\includegraphics[width=0.5\textwidth]{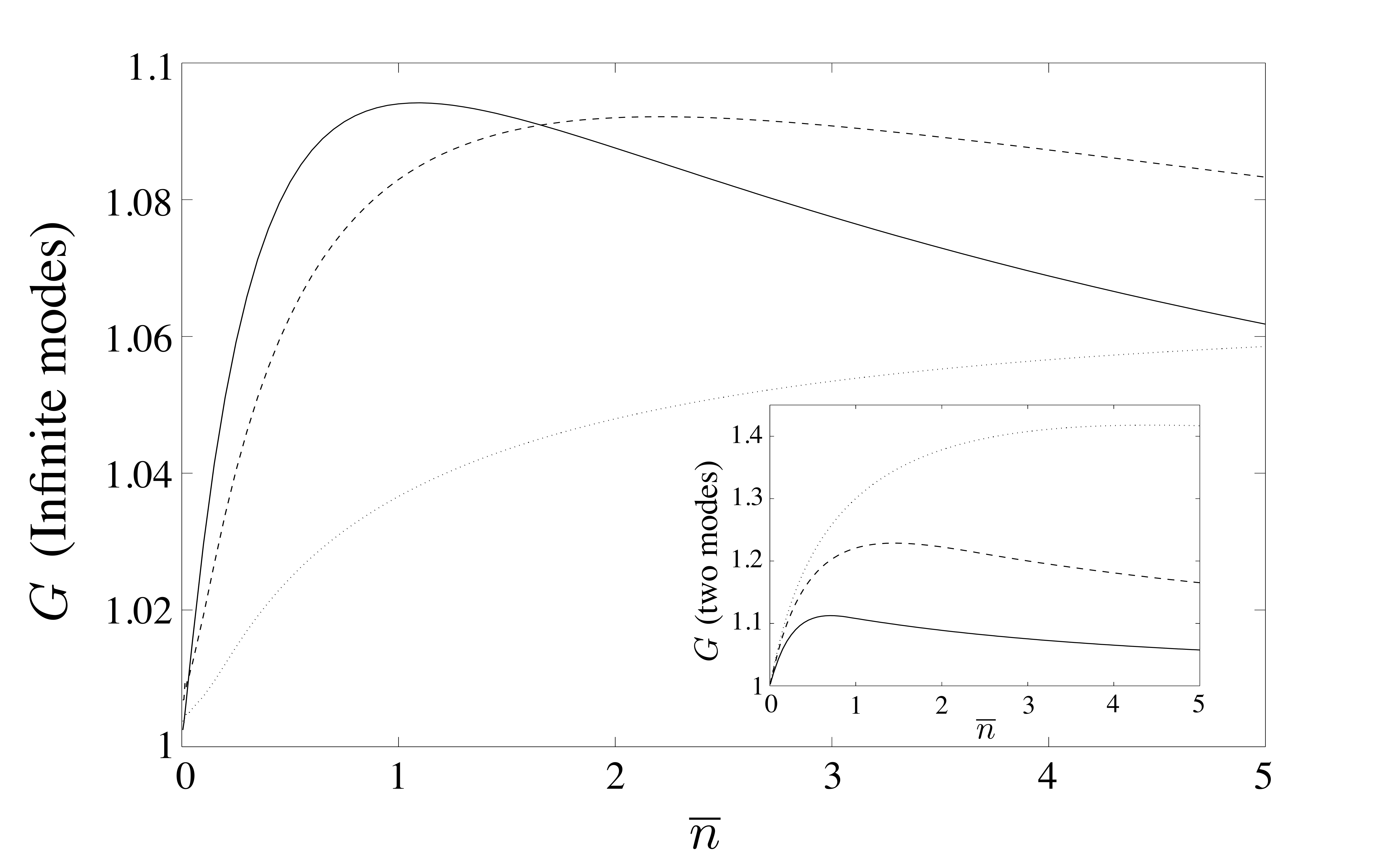}		
	\caption{Gain $G$ vs. $\overline{n}$ for an infinite number of modes and for two modes (inset). For both plots we took SNR = 3 and $\phi = 0.7$ (solid curve), $\phi = 0.9$ (dashed curve), $\phi = 0.99$ (dotted curve).}
	\label{fig:markGainCompvsnmax}
\end{figure}
\begin{figure}
	\centering
		\includegraphics[width=0.5\textwidth]{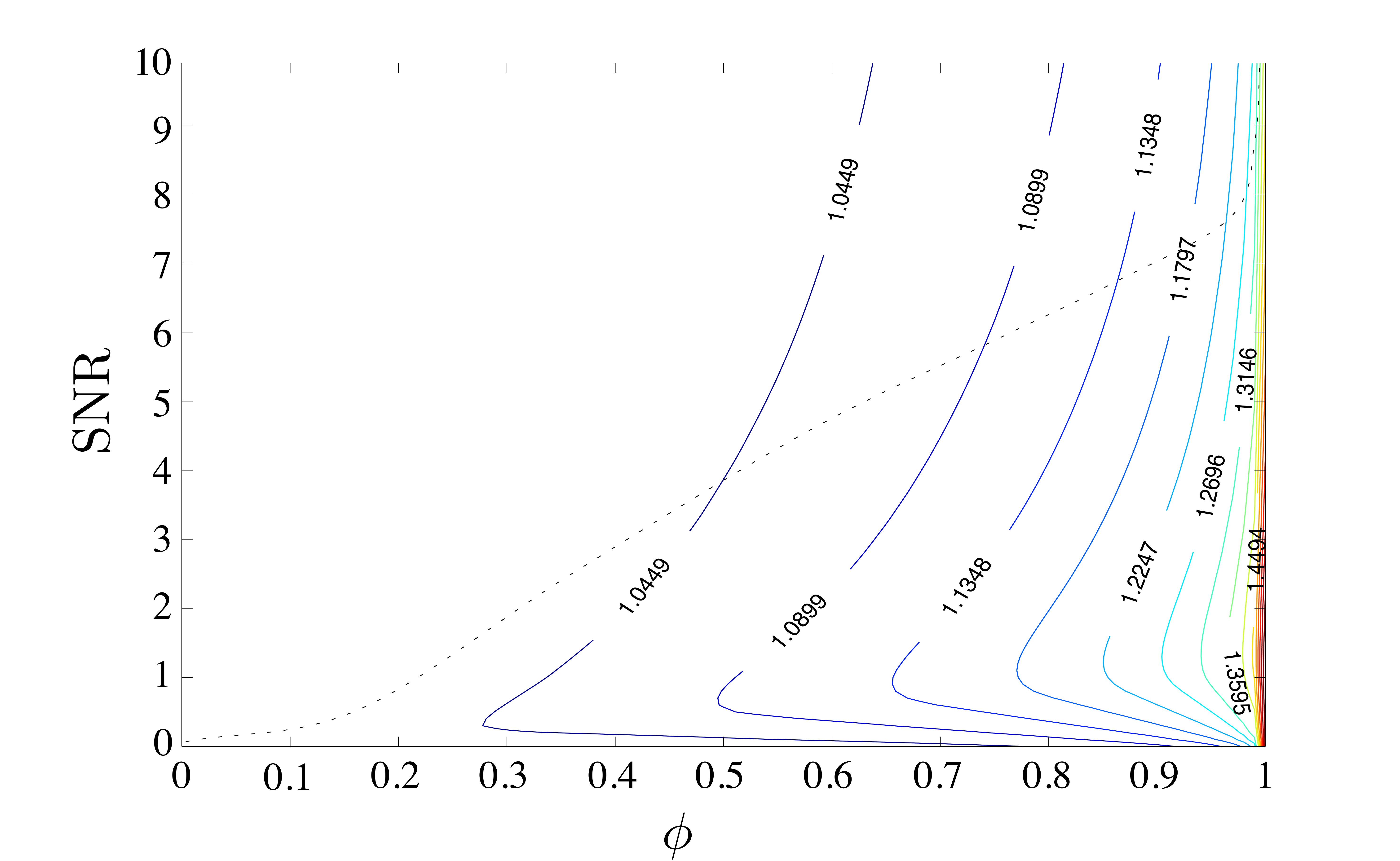}
	\caption{(Color online) Contour plot of the maximal gain $\max_{\overline{n}}G$ vs. $\phi$, SNR for two modes. In the area above the dotted line the quantum water-filling solution holds.}
	\label{fig:monGainContourphiSNR}
\end{figure}
\begin{figure}
	\centering
		\includegraphics[width=0.5\textwidth]{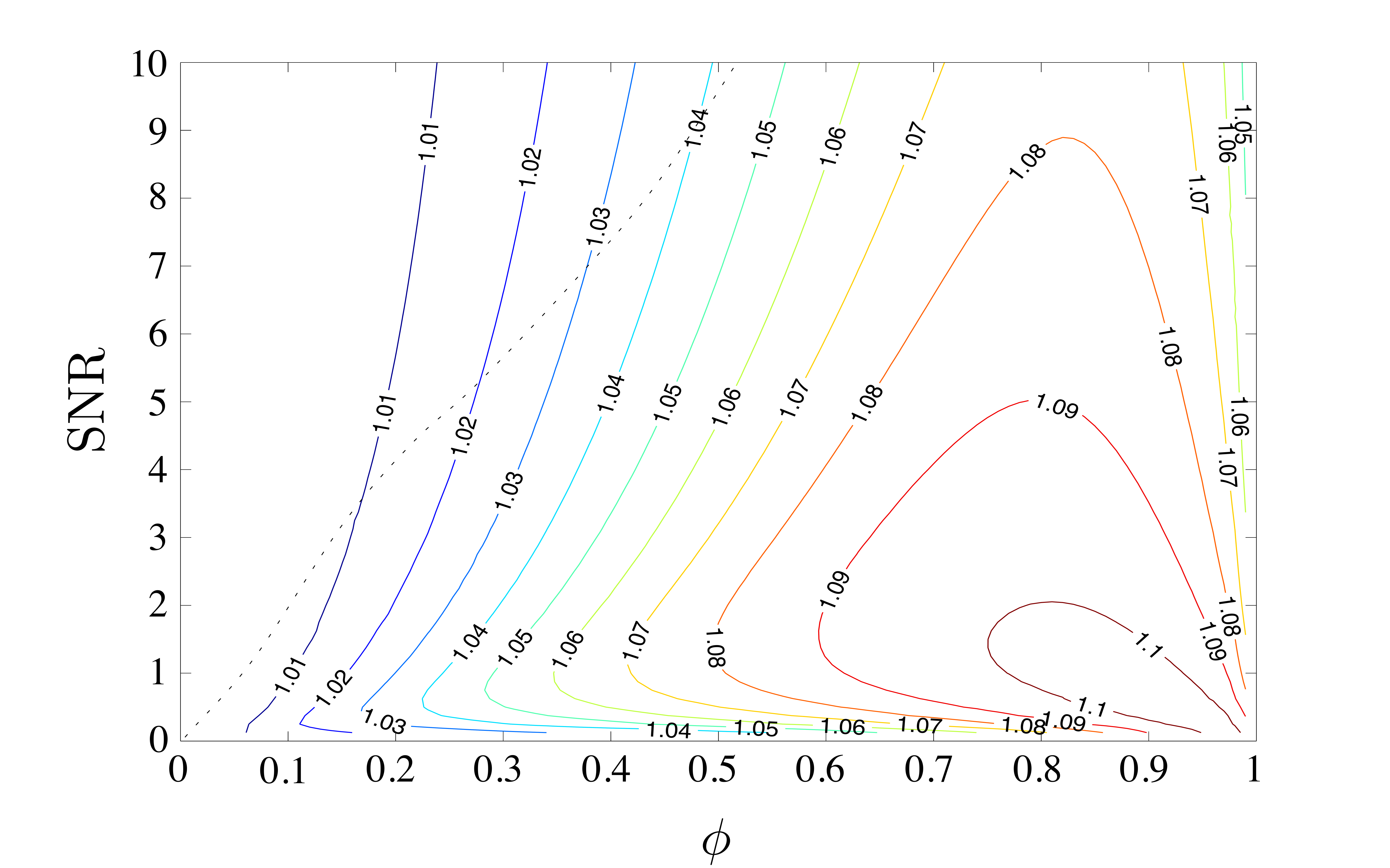}
	\caption{(Color online) Contour plot of the maximal gain $\max_{\overline{n}}G$ vs. $\phi$, SNR for an infinite number of modes. In the area left to the dotted line the global quantum water-filling solution holds.}
	\label{fig:markGainContourphiSNR}
\end{figure}
we plotted the gain $G$ vs. $\overline{n}$ for fixed SNR and different $\phi$ for an infinite number of modes and for two modes. We see that unlike in the case of two modes, where the gain with higher correlations is always higher, in the case of an infinite number of modes the maximum of the gain is found for some intermediate correlations. However, in this plot one does not see the dependency on the SNR. So the question that follows is: What is the dependence of the maximal gain (with respect to $\overline{n}$) on $\phi$ and the SNR? In order to answer this question we make a contour plot of $\max_{\overline{n}}{G}$ vs. $\phi$ and SNR. In the case of two modes we see in Fig.~\ref{fig:monGainContourphiSNR} that the optimal gain is obtained at $\phi = 1$ for a certain SNR. In addition, in this case, the increase in gain at high correlations is very strong compared to that at lower correlations. Furthermore, a low SNR seems to benefit from entanglement more than a higher SNR.

For an infinite number of modes the situation is different, as we can see in Fig.~\ref{fig:markGainContourphiSNR}: instead of a sharp edge toward high correlations we see an almost-flat area of maximal gain in the region of high correlations and low SNR. This holds, on one hand, for a high correlation and low SNR but, on the other hand, also for less correlated noise and a higher SNR. Furthermore, the enhancement is rather robust and does not drop as sharply with decreasing correlations as in the two- mode case. However, as the region of high gain has input energies below the global water-filling threshold, where the optimal input squeezing becomes quite complex [as depicted, e.g., in Fig.~\ref{fig:optinmod} (b)] a modulation of coherent states might be practically more favorable, because it is already quite efficient as the gain due to entanglement does not exceed 10\%.

%%%%%%%%%%%%%%%%%%%%%%%%%%%%%%%%%%%%%%%%%%%%%%%%%%
\section{Conclusions}\label{sec:conclusions}
%%%%%%%%%%%%%%%%%%%%%%%%%%%%%%%%%%%%%%%%%%%%%%%%%%
We have presented an algorithm for calculation of the classical Gaussian capacity of the Gaussian channel with additive correlated noise. This method is applicable to all cases where there are no $q-p$ correlations in the noise covariance matrix and, moreover, the $q$ and $p$ blocks commute at least asymptotically in the limit of an infinite number of uses. This applies, in particular, to the whole class of channels in which noise correlations are given by a stationary Gauss process.

We applied this method to a channel with a Gauss-Markov noise as defined in \cite{SDKC09}, which has asymptotically commuting block matrices. We found that in the limit of full correlations the capacity tends to the noiseless capacity. We calculated the covariance matrix of the optimal input state not only in the eigenbasis of the noise covariance matrix but also in the original, correlated basis. In this correlated basis the optimal input state is entangled and we found that the degree of entanglement scales with the correlation parameter of the noise from no entanglement (i.e., a set of coherent states) to a maximally entangled state.

Furthermore, we discussed the gain from using optimal entangled input states with respect to coherent product states in the case of two modes and an infinite number of modes. We found that, contrarily to the two-mode case, where the gain always strongly increases with correlations for any SNR, for an infinite number of modes a high gain is achieved in a region of high correlations and low SNR and of lower correlations and higher SNR. In addition, the gain in the limit of an infinite number of modes does not drop as sharply with decreasing correlations as in the two-mode case. We also observed that a Gaussian modulated coherent-state encoding already achieves not less than 90\% of the Gaussian capacity.

\begin{acknowledgements}
The authors acknowledge the financial support from the EU under the FP7 project COMPAS,
from the Belgian federal government via the IAP research network Photonics$@$be,
from the Brussels Capital Region under the project CRYPTASC and from
the Belgian foundation FRIA. We thank Oleg Pilyavets for fruitful discussions and many comments on the paper and Stefano Mancini for pointing out useful references.
\end{acknowledgements}

\appendix
%%%%%%%%%%%%%%%%%%%%%%%%%%%%%%%%%%%%%%%%%%%%%%%%%%
\section{Results for one mode}\label{sec:lagrange}
%%%%%%%%%%%%%%%%%%%%%%%%%%%%%%%%%%%%%%%%%%%%%%%%%%
In the following we present the solution via the Lagrange multipliers method for the optimization problem for the one-mode channel introduced in section \ref{sec:optproblem}.
%%%%%%%%%%%%%%%%%%%%%%%%%%%%%%%%%%%%%%%%%%%%%%%%%%
\subsection{Search for the extremum}\label{sec:lagrangeqwf}
%%%%%%%%%%%%%%%%%%%%%%%%%%%%%%%%%%%%%%%%%%%%%%%%%%
The extremum of the Lagrangian ${\mathcal L}$ defined in \eqref{eq:Lone} must satisfy
\[
	\nabla \, {\mathcal L} = 0,
\]
where
\[
	\nabla = \left(\frac{\partial}{\partial\gamma^{q}_{\rm in}},\frac{\partial}{\partial\gamma^{p}_{\rm in}},\frac{\partial}{\partial\gamma^{q}_{\rm mod}},\frac{\partial}{\partial\gamma^{p}_{\rm mod}},\frac{\partial}{\partial\gamma^{qp}_{\rm in}},\frac{\partial}{\partial\gamma^{qp}_{\rm mod}}\right)^\mathrm{T}.
\]
This corresponds to six equations:
\begin{eqnarray}
	\kappa(\overline{\nu}) \, \overline{\gamma}^p - \kappa(\nu_{\rm out}) \, \gamma^p_{\rm out} - \mu - \tau \gamma_{\rm in}^p & = & 0,\label{eq:6eqinq}\\
	\kappa(\overline{\nu}) \, \overline{\gamma}^q - \kappa(\nu_{\rm out}) \, \gamma^q_{\rm out} - \mu - \tau \gamma_{\rm in}^q & = & 0,\label{eq:6eqinp}\\
	\kappa(\overline{\nu}) \, \overline{\gamma}^p - \mu & = & 0,\label{eq:6eqmodq}\\
	\kappa(\overline{\nu}) \, \overline{\gamma}^q - \mu & = & 0,\label{eq:6eqmodp}\\
	-\kappa(\overline{\nu}) \, (\gamma_{\rm in}^{qp} + \gamma_{\rm mod}^{qp}) + \kappa(\nu_{\rm out}) \, \gamma_{\rm in}^{qp} + \tau \gamma_{\rm in}^{qp} & = & 0,\label{eq:6eqinqp}\\
	-\kappa(\overline{\nu}) \, (\gamma_{\rm in}^{qp} + \gamma_{\rm mod}^{qp}) & = & 0 \label{eq:6eqmodqp},
\end{eqnarray}
where $\mu, \tau$ are Lagrange multipliers and
\[
	\kappa(x) = \frac{g'(x-\frac{1}{2})}{2 x}.
\]
From Eqs.~\eqref{eq:6eqmodq} and \eqref{eq:6eqmodp} we derive Eqs.~\eqref{eq:oneqwf} and \eqref{eq:oneqwfmu}. Since $\kappa(x) > 0$ for all $x > 1/2$ we find from Eq.~\eqref{eq:6eqmodqp} that $\gamma_{\rm in}^{qp} = -\gamma_{\rm mod}^{qp}$. Therefore, Eq.~\eqref{eq:6eqinqp} simplifies to
\begin{equation}
	\gamma_{\rm in}^{qp}[\kappa(\nu_{\rm out}) + \tau] = 0.
\end{equation}
If one assumes that $\gamma_{\rm in}^{qp} \neq 0$ then the resulting multiplier $\tau$ leads to a contradiction when inserted in Eqs.~\eqref{eq:6eqinq} and \eqref{eq:6eqinp}. Thus, we conclude that $\gamma_{\rm in}^{qp} = \gamma_{\rm mod}^{qp} = 0$. Finally, by combining equations \eqref{eq:6eqinq}-\eqref{eq:6eqmodq} one deduces Eq.~\eqref{eq:oneoptin}. 

In order for the solutions to be physical the modulation eigenvalues $\gamma_{\rm mod}^{q}, \gamma_{\rm mod}^{p}$ have to be positive, which is the case for an input energy above the energy threshold \eqref{eq:onethr}. For such $\lambda$ we will show now that ${\mathcal L}$ is concave on the solution which proves that we found indeed a local maximum. For fixed $\lambda$, we can prove concavity by checking whether the Hessian $H({\mathcal{L}})$, i.e. the $6 \times 6$ matrix that contains the second derivates of $\mathcal L$ with respect to all variables $\gamma^{q,p}_{\rm in},\gamma^{q,p}_{\rm mod},\gamma^{qp}_{\rm in,mod}$, is negative definite. First, we find that
\[
	  \frac{\partial^2{\mathcal L}}{\partial \gamma^{qp}_{\rm s}\partial \gamma^{r}_{\rm s}} \Big |_{\gamma^{qp}_{\rm in}=0,\gamma^{qp}_{\rm mod}=0} = 0, \quad r=q,p; \; {\rm s = \rm in, \rm mod},
\]
which correspond to the derivatives of the left hand sides of Eqs.~\eqref{eq:6eqinqp} and \eqref{eq:6eqmodqp} with respect to $\gamma^{q,p}_{\rm in},\gamma^{q,p}_{\rm mod}$. Therefore, we cab write the Hessian $H({\mathcal L})$ in the block form
\[
	H({\mathcal{L}}) = 
	\begin{pmatrix}
	H_{\rm var}({\mathcal{L}}) & 0\\
	0 & H_{\rm cov}({\mathcal{L}})
	\end{pmatrix},
\]
where $H_{\rm var}({\mathcal{L}})$ contains only second derivatives with respect to $\gamma^{q,p}_{\rm in},\gamma^{q,p}_{\rm mod}$ and $H_{\rm cov}({\mathcal{L}})$ second derivatives with respect to $\gamma^{qp}_{\rm in,mod}$. By using $\gamma_{\rm mod}^{qp}=\gamma_{\rm in}^{qp}=0$ and Eqs.~\eqref{eq:6eqinq}, \eqref{eq:6eqinp} we find 
\begin{equation}\label{eq:Hcov}
	H_{\rm cov}({\mathcal{L}}) = -
	\begin{pmatrix}
	A+B & A\\
	A & A
	\end{pmatrix},
\end{equation}
where $A = \kappa(\overline{\nu}), B = \kappa(\nu_{\rm out})c$, with
\begin{equation}
	c = 2\sqrt{\gamma_{\rm env}^q\gamma_{\rm env}^p} \, \frac{\gamma_{\rm env}^q + \gamma_{\rm env}^p}{\gamma_{\rm env}^q - \gamma_{\rm env}^p}.
\end{equation}
The eigenvalues of \eqref{eq:Hcov} read
\begin{equation}\label{eq:hqp12}
	h^{qp}_{1,2} = -A -\frac{B}{2} \pm \sqrt{A^2 + \frac{B^2}{4}},
\end{equation}  
which are both negative, since $A, B > 0$.

Now we show that the eigenvalues of $H_{\rm var}({\mathcal{L}})$ are also negative. However, instead of considering the Hessian $H_{\rm var}({\mathcal L})$ we consider equivalently the Hessian $H_{\rm var}({\mathcal \chi})$, where we embedded constraints \eqref{eq:purity}, \eqref{eq:onelambda} in $\chi$, which becomes then a function of only two variables $\gamma_{\rm in}^q,\gamma_{\rm mod}^q$ and thus $H_{\rm var}({\mathcal \chi})$ is a $2 \times 2$ matrix. Then we find that $H_{\rm var}({\mathcal \chi})$ has the same shape as \eqref{eq:Hcov}, where now
\begin{equation}
	A = \frac{g'\left(\overline{\nu} - \frac{1}{2} \right)}{\overline{\nu}}, \quad B = \frac{g'\left(\nu_{\rm out} - \frac{1}{2} \right)\gamma_{\rm env}^q}{\nu_{\rm out}4(\gamma_{\rm in}^q)^3}.
\end{equation}
Since $A,B > 0$ we conclude again that both eigenvalues [that read like \eqref{eq:hqp12}] are negative. Therefore, the total Hessian $H({\mathcal{L}})$ on the solution is negative definite, which proves the concavity of ${\mathcal{L}}$ at the extremal point for an input energy above the threshold \eqref{eq:onethr}. Thus, we conclude that we have found a local maximum of ${\mathcal{L}}$.

%%%%%%%%%%%%%%%%%%%%%%%%%%%%%%%%%%%%%%%%%%%%%%%%%%
\subsection{Below the threshold}\label{sec:lagrangenowf}
%%%%%%%%%%%%%%%%%%%%%%%%%%%%%%%%%%%%%%%%%%%%%%%%%%
If the input energy is below the threshold \eqref{eq:onethr} then the extremum of the Lagrangian lays outside the physical region. In this case the maximum lays on its boundary which corresponds to one or both modulation eigenvalues being 0. If both modulation eigenvalues are 0 no information is transmitted which is clearly not the optimum for an input energy $\lambda > 1$. 
Now we put $\gamma_{\rm mod}^q = 0$ which replaces Eq.~\eqref{eq:6eqmodq} (if we put instead $\gamma_{\rm mod}^p = 0$ we would obtain in the following the same set of equations up to a permutation of indexes $p$ and $q$). Therefore, this degree of freedom does no longer exist and there is no $q-p$ correlation in the modulation: $\gamma_{\rm mod}^{qp} = 0$. Now, we have to find the extremum of the Lagrangian again, where the new gradient reads
\[
	\nabla = \left(\frac{\partial}{\partial\gamma^{q}_{\rm in}},\frac{\partial}{\partial\gamma^{p}_{\rm in}},\frac{\partial}{\partial\gamma^{p}_{\rm mod}},\frac{\partial}{\partial\gamma^{qp}_{\rm in}}\right)^\mathrm{T}.
\]
The previous equation, \eqref{eq:6eqinqp}, is now simplified to
\[
	-\gamma_{\rm in}^{qp}[\kappa(\overline{\nu}) - \kappa(\nu_{\rm  out}) - \tau] = 0.
\]
Again, if one takes $\gamma_{\rm in}^{qp} \neq 0$ then the solution for $\tau$ when inserted into Eqs.~\eqref{eq:6eqinq} and \eqref{eq:6eqinp} leads to a contradiction, namely $\gamma_{\rm mod}^p < 0$, which is unphysical. Therefore, we conclude that 
\begin{equation}\label{eq:ginqp}
	\gamma_{\rm in}^{qp} = 0.
\end{equation}

By inserting \eqref{eq:6eqmodp} into \eqref{eq:6eqinq} one obtains $\tau$ which when inserted into \eqref{eq:6eqinp} leads to the transcendental equation \eqref{eq:onebt}. The Lagrange multiplier $\mu$ is given by \eqref{eq:6eqmodp} leading to Eq.~\eqref{eq:onebtmu}.

From Eq.~\eqref{eq:onebt} we can deduce a bound on the optimal input squeezing. From the fact that
\[
	\frac{g'(\overline{\nu} - \frac{1}{2})}{\overline{\nu}} \leq \frac{g'(\nu_{\rm out} - \frac{1}{2})}{\nu_{\rm out}},
\]
as $1/x \, g'(x-1/2)$ is a monotonically decreasing function and $\overline{\nu} > \nu_{\rm out}$, it follows that 
\begin{equation}\label{eq:ineq1}
	|\overline{\gamma}^p - \overline{\gamma}^q| \ge |\gamma_{\rm env}^p - \frac{\gamma_{\rm env}^q}{4 (\gamma_{\rm in}^q)^2}|,
\end{equation}
where the equal sign holds if both left and right hand side are 0 and the expressions inside the absolute values on both sides have the same sign. When $\lambda < \lambda_{\rm thr}$ they cannot be 0, because otherwise $\gamma_{\rm in}^q = 1/2\sqrt{\gamma_{\rm env}^q/\gamma_{\rm env}^p}$ and $\overline{\gamma}^p = \overline{\gamma}^q$ leading to $\lambda = \lambda_{\rm thr}$ which contradicts to our assumption that we are below the threshold. This proves that
\begin{equation}\label{eq:gbarqnegbarp}
	\overline{\gamma}^q \ne \overline{\gamma}^p.
\end{equation}

Let us remind that up to here the obtained results are invariant (up to the permutation of indices $q$ and $p$) with respect to the choice which modulation eigenvalue is set to 0.

Now we prove that the choice which modulation eigenvalue has to be set to 0 is determined by the relation between $\gamma_{\rm env}^{q}$ and $\gamma_{\rm env}^{p}$.
\begin{lemmodq}\label{lemmmodq}
	For an input energy $\lambda < \lambda_{\rm thr}$ under the condition $\gamma_{\rm env}^{q} > \gamma_{\rm env}^{p}$ the maximum of $\chi$ given by Eq.~\eqref{eq:chi} is achieved when 
	\begin{equation}
		\gamma_{\rm mod}^q = 0.
	\end{equation}
\end{lemmodq}
\begin{proof}
Let us assume that in contrary to the statement of the lemma we have $\gamma_{\rm mod}^p = 0$. We know that due to Eq.~\eqref{eq:gbarqnegbarp} we have two possible cases: $\overline{\gamma}^q > \overline{\gamma}^p$ or $\overline{\gamma}^q < \overline{\gamma}^p$. 

Let us assume first that $\overline{\gamma}^q > \overline{\gamma}^p$ and that we have found an optimal solution. Suppose we remove a fraction of $\gamma_{\rm mod}^q$ which is smaller then half of the difference $\overline{\gamma}^q - \overline{\gamma}^p$ and set $\gamma_{\rm mod}^p$ equal to this fraction. This will not change the input energy as well as the output entropy (second term in $\chi$). However, the overall modulated output entropy (first term in $\chi$) will increase. The reason is that for a constant arithmetic mean of $a$ and $b$ the geometric mean increases when the difference between $a$ and $b$ decreases. Therefore, by doing this we increased $\chi$ which leads to a contradiction because we the assumed optimal solution is in fact not optimal. Thus, for $\overline{\gamma}^q > \overline{\gamma}^p$ the lemma is proven.

Now we consider $\overline{\gamma}^q < \overline{\gamma}^p$ and assume that we have found an optimal solution. Then from Eq.~\eqref{eq:ineq1} it follows that
\begin{equation}
	\gamma_{\rm in}^q > \frac{1}{2} \sqrt{\frac{\gamma_{\rm env}^q}{\gamma_{\rm env}^p}}
\end{equation}
and from the condition $\gamma_{\rm env}^{q} > \gamma_{\rm env}^{p}$ we deduce easily that $\gamma_{\rm in}^q > 1/2 > \gamma_{\rm in}^p$. Taking into account that $\gamma_{\rm mod}^q > \gamma_{\rm mod}^p = 0$ we conclude that in fact $\overline{\gamma}^q > \overline{\gamma}^p$ which contradicts to our assumption. 

Thus, the lemma is proven.
\end{proof}

A direct consequence of Lemma \ref{lemmmodq} and Eq.~\eqref{eq:ineq1} is that 
\begin{equation}\label{eq:ginq}
	\frac{1}{2} \leq \gamma_{\rm in}^q < \frac{1}{2} \sqrt{\frac{\gamma_{\rm env}^q}{\gamma_{\rm env}^p}} \equiv \gamma_{\rm in (thr)}^{q},
\end{equation}
meaning that the squeezing of the input state for $\lambda < \lambda_{\rm thr}$ is always smaller than for $\lambda \ge \lambda_{\rm thr}$.

The next step is to prove that the extremum we found is a maximum. In order to do this we verify that ${\mathcal L}$ is concave for an input energy below the threshold \eqref{eq:onethr}, where $\lambda$ is fixed. First, we see again that all cross derivatives with respect to $\gamma_{\rm in}^{qp}$ and the three variances vanish on the solution. Therefore, we consider in the Hessian $H({\mathcal L})$ the second derivative with respect to $\gamma_{\rm in}^{qp}$ separately. By using Eqs.~\eqref{eq:6eqinp} and \eqref{eq:6eqmodp} we find
\begin{equation}
	 \frac{\partial^2 {\mathcal L}}{\partial (\gamma_{\rm in}^{qp})^2} = -\kappa(\overline{\nu}) - \kappa(\nu_{\rm out})\left(\frac{\gamma_{\rm out}^q}{\gamma_{\rm in}^q} - 1\right) < 0,
\end{equation} 
since $\gamma_{\rm out}^q > \gamma_{\rm in}^q$.

Again, for the remaining three variables instead of considering the Hessian of $({\mathcal L})$ we consider equivalently the Hessian of $\chi$ with constraints \eqref{eq:purity} and \eqref{eq:onelambda} embedded. Now, $\chi$ becomes a function of only one variable $\gamma_{\rm in}^q$, and thus the Hessian is only the second derivative of $\chi$ with respect to $\gamma_{\rm in}^q$. We prove in Appendix \ref{sec:mulambda} that
\begin{equation}
	\frac{\partial^2 \chi}{\partial (\gamma_{\rm in}^q)^2} = \frac{\partial F}{\partial \gamma_{\rm in}^q} < 0,
\end{equation}
where $F$ is given by \eqref{eq:F}. Therefore, the full Hessian of ${\mathcal L}$ is negative definite for input energies below the threshold as well. Thus, we have shown that ${\mathcal L}$ is concave on the solution, which proves that we indeed found a local maximum of ${\mathcal L}$.

%%%%%%%%%%%%%%%%%%%%%%%%%%%%%%%%%%%%%%%%%%%%%%%%%%
\subsection{Bounds below the threshold}\label{sec:bounds}
%%%%%%%%%%%%%%%%%%%%%%%%%%%%%%%%%%%%%%%%%%%%%%%%%%
In the following we deduce bounds on the overall modulated output variances for $\lambda < \lambda_{\rm thr}$. From Eqs.~\eqref{eq:ineq1} and \eqref{eq:ginq} it follows directly that 
\begin{equation}\label{eq:ineq2}
	\overline{\gamma}^p < \overline{\gamma}^q.
\end{equation}

Furthermore, we can find a lower bound on $\overline{\gamma}^p$. Suppose $\overline{\gamma}^p < 1/2$, then we have
\begin{equation}
	\begin{split}
		\frac{1}{\overline{\nu}} = \frac{1}{\sqrt{\overline{\gamma}^q\overline{\gamma}^p}} & > \frac{1}{\sqrt{\frac{1}{2}\overline{\gamma}^q}}\\
		\Rightarrow g'\left(\overline{\nu}-\frac{1}{2}\right) & > g'\left(\sqrt{\frac{1}{2} \overline{\gamma}^q}-\frac{1}{2}\right),
	\end{split}
\end{equation}
since $g'(x-1/2)$ is a monotonically decreasing function of $x$. Thus, we can rewrite Eq.~\eqref{eq:onebt},
\begin{equation}
		\frac{g'\left(\overline{\nu} - \frac{1}{2}\right)}{\overline{\nu}}\left(\overline{\gamma}^q - \overline{\gamma}^p\right) > \frac{g'\left(\sqrt{\frac{1}{2}\overline{\gamma}^q} - \frac{1}{2}\right)}{\sqrt{\frac{1}{2}\overline{\gamma}^q}}\left(\overline{\gamma}^q - \frac{1}{2}\right)
\end{equation}
and by using \eqref{eq:onebt} and the fact that below the threshold $\overline{\gamma}^q = \gamma_{\rm out}^q$ we find
\begin{equation}\label{eq:ineqgbarp}
	\frac{g'(\nu_{\rm out} - \frac{1}{2})}{2 \nu_{\rm out}} \Sigma > \frac{g'\left(\sqrt{\frac{1}{2}\gamma_{\rm out}^q} - \frac{1}{2}\right)}{\sqrt{\frac{1}{2}\gamma_{\rm out}^q}}\left(\gamma_{\rm out}^q - \frac{1}{2}\right)
\end{equation} 
where
\begin{equation}
	\Sigma = \frac{\gamma_{\rm env}^q}{4(\gamma_{\rm in}^q)^2} - \gamma_{\rm env}^p.
\end{equation}
Thus, our assumption $\overline{\gamma}^p < 1/2$ leads to an inequality which depends solely on $\gamma_{\rm in}^q, \gamma_{\rm env}^q, \gamma_{\rm env}^p$ with the constraints on $\gamma_{\rm in}^q$ given by \eqref{eq:ginq} and for the noise variances that $\gamma_{\rm env}^q > \gamma_{\rm env}^p$, $0 \leq \gamma_{\rm env}^p \leq 1/2-1/(4\gamma_{\rm in}^q)$. Since, inequality \eqref{eq:ineqgbarp} is always violated for the given constraints, we come to a contradiction which proves that 
\begin{equation}\label{eq:gbarp}
	\overline{\gamma}^p \geq 1/2.
\end{equation}

%%%%%%%%%%%%%%%%%%%%%%%%%%%%%%%%%%%%%%%%%%%%%%%%%%
\subsection{Monotonicity of $\mu$}\label{sec:mulambda}
%%%%%%%%%%%%%%%%%%%%%%%%%%%%%%%%%%%%%%%%%%%%%%%%%%
\begin{lemmu} 
	The Lagrange multiplier $\mu$ is a monotonically decreasing function of the input energy $\lambda$ on the solution, and moreover
	\begin{equation}\label{eq:dmudlambda}
		\frac{d \mu}{d \lambda} < 0.
	\end{equation}	
\end{lemmu}
\begin{proof}
For $\lambda \geq \lambda_{\rm thr}$ the proof follows directly from the definition of $\mu$ by Eq.~\eqref{eq:oneqwfmu}, the definition of $\overline{\nu}_\mathrm{wf}$ by Eq.~\eqref{eq:nuwf} and the fact that $g'(x)$ is a monotonically decreasing function.
	
Now we prove the lemma For $\lambda < \lambda_{\rm thr}$. Using Eq.~\eqref{eq:onebtmu} we can write
\begin{equation}
	\frac{d \mu}{d \lambda} = g''\left(\overline{\nu}-\frac{1}{2}\right)\frac{d \bar{\nu}}{d \lambda}\frac{\overline{\gamma}^q}{2\overline{\nu}} + \frac{g'\left(\nu - \frac{1}{2}\right)}{2\overline{\nu}^2}\left(\frac{d\overline{\gamma}^q}{d\lambda}\overline{\nu} - \overline{\gamma}^q\frac{d\overline{\nu}}{d\lambda}\right).
\end{equation}
We can upper bound this quantity if we use the following property of $g(x)$:
\begin{equation}
	-\frac{g'(\overline{\nu}-1/2)}{\overline{\nu}} > g''(\overline{\nu}-1/2).
\end{equation}
This leads to
\begin{equation}\label{eq:dmudlambdaub}
	\frac{d \mu}{d \lambda} < -\frac{g'\left(\overline{\nu} - \frac{1}{2}\right)}{2\overline{\nu}^3}(\overline{\gamma}^q)^2 \frac{d\overline{\gamma}^p}{d\lambda}.
\end{equation} 
Since all factors in \eqref{eq:dmudlambdaub} except for $d\overline{\gamma}^p/d\lambda$ are clearly positive the lemma will be proven if we show that
\begin{equation}\label{eq:dgbarpdlambda}
	\frac{d\overline{\gamma}^p}{d\lambda} = 1 - \frac{d\gamma_{\rm in}^q}{d\lambda}> 0.
\end{equation}
This derivative can be expressed in terms of the function
\begin{equation}\label{eq:F}
	F \equiv \frac{g'(\overline{\nu} - \frac{1}{2})}{2 \overline{\nu}}\left({\overline{\gamma}}^p - {\overline{\gamma}}^q\right) - \frac{g'(\nu_{\rm out} - \frac{1}{2})}{2 \nu_{\rm out}}\left(\gamma_{\rm out}^p - \frac{\gamma_{\rm in}^p}{\gamma_{\rm in}^q}\gamma_{\rm out}^q\right).
\end{equation}
Indeed, when Eq.~\eqref{eq:onebt} holds then we have
\begin{equation}\label{eq:dinqdlambda}
	\frac{d\gamma_{\rm in}^q}{d\lambda} = - \frac{\frac{\partial F}{\partial \lambda}}{\frac{\partial F}{\partial \gamma_{\rm in}^q}}.
\end{equation}
% Therefore, inequality Eq.~\eqref{eq:dgbarpdlambda} is equivalent to
% \begin{equation}
% 	|\frac{\partial F}{\partial \lambda}| - |\frac{\partial F}{\partial \gamma_{\rm in}^q}| < 0.
% \end{equation}
We observe that
\begin{equation}\label{eq:dFdlambda}
	\begin{split}
	\frac{\partial F}{\partial \lambda} & = \frac{g''\left(\overline{\nu}-\frac{1}{2}\right)}{4\overline{\nu}^2}\overline{\gamma}^q(\overline{\gamma}^p-\overline{\gamma}^q)\\
	& + \frac{g'\left(\overline{\nu} - \frac{1}{2}\right)}{4\overline{\nu}}\left(1+\frac{(\overline{\gamma}^q)^2}{\overline{\nu}^2}\right) > 0,
	\end{split}
\end{equation}
because $g''(x) < 0$, $g'(x) > 0$ and $\overline{\gamma}^q > \overline{\gamma}^p$. Thus, in order to prove inequality Eq.~\eqref{eq:dgbarpdlambda} it suffices to prove that
\begin{equation}\label{eq:dFdlambdadinq}
	\frac{\partial F}{\partial \lambda} + \frac{\partial F}{\partial \gamma_{\rm in}^q} < 0.
\end{equation}
By carrying out the partial derivatives we rewrite Eq.~\eqref{eq:dFdlambdadinq} in the form 
\begin{equation}\label{eq:6terms}
	-\frac{\eta}{4 \overline{\nu}^3}\overline{\gamma}^p(\overline{\gamma}^q -\overline{\gamma}^p) T_1 - \frac{1}{4 \nu_{\rm out}^3} T_2 < 0,
\end{equation}
where $T_1$ and $T_2$ denote the expressions
\begin{equation}
	\begin{split}
		T_1 & = g''\left(\overline{\nu} - \frac{1}{2}\right)\frac{\overline{\nu}}{\eta} + g'\left(\overline{\nu}-\frac{1}{2}\right),\\
		T_2 & = g''\left(\nu_{\rm out} - \frac{1}{2}\right) \nu_{\rm out} \, \zeta\\
		 	& + g'\left(\nu_{\rm out} - \frac{1}{2}\right) \left(\frac{\gamma_{\rm env}^q}{(\gamma_{\rm in}^q)^3}\nu_{\rm out}^2 - \zeta\right),
	\end{split}
\end{equation}
and
\begin{equation}
	\eta = \frac{\overline{\gamma}^q + \overline{\gamma}^p}{\overline{\gamma}^q - \overline{\gamma}^p}, \; \zeta = \left(\gamma_{\rm env}^p - \frac{\gamma_{\rm env}^q}{4 (\gamma_{\rm in}^q)^2}\right)^2.
\end{equation}
Observe that all factors in front of $T_1$ and $T_2$ are positive, since $\overline{\gamma}^{p},\overline{\nu},\nu_{\rm out},\overline{\gamma}^q-\overline{\gamma}^p > 0$. If we prove positive of $T_1$ and $T_2$, then inequality \eqref{eq:6terms} will be proven as well as the lemma.

The positivity of $T_1$ can be verified via its partial derivatives with respect to $\overline{\gamma}^q, \overline{\gamma}^p$ which lead to:
\begin{equation}
	\frac{\partial T_1}{\partial \overline{\gamma}^q} = \{\overline{\gamma}^q + \overline{\gamma}^p [3 - 4\overline{\gamma}^p(\overline{\gamma}^p+3\overline{\gamma}^q)]\}  {T_{11}}(\overline{\gamma}^q,\overline{\gamma}^p), \label{eq:dT1dq}
\end{equation} 
\begin{equation}
	\frac{\partial T_1}{\partial \overline{\gamma}^p} = (\overline{\gamma}^q - \overline{\gamma}^p) \, {T_{12}}(\overline{\gamma}^q,\overline{\gamma}^p), \label{eq:dT1dp}
\end{equation}
where
\begin{equation}
	\begin{split}
		T_{11}(\overline{\gamma}^q,\overline{\gamma}^p) & = \frac{4\sqrt{\overline{\gamma}^q\overline{\gamma}^p}}{(\overline{\gamma}^q+\overline{\gamma}^p)^3(1-4\overline{\gamma}^q\overline{\gamma}^p)^2}\\ 
		T_{12}(\overline{\gamma}^q,\overline{\gamma}^p) & = \frac{4(\overline{\gamma}^q)^2(1+4(\overline{\gamma}^p)^2)}{\sqrt{\overline{\gamma}^q\overline{\gamma}^p}(\overline{\gamma}^q+\overline{\gamma}^p)^2(1-4\overline{\gamma}^q\overline{\gamma}^p)^2}.
	\end{split}	
\end{equation} 
Clearly the two functions ${T_{11}}(\overline{\gamma}^q,\overline{\gamma}^p)$ and ${T_{12}}(\overline{\gamma}^q,\overline{\gamma}^p)$ are positive for all $\overline{\gamma}^{q,p} > 0$. Then \eqref{eq:dT1dp} is also positive since $\overline{\gamma}^q-\overline{\gamma}^p > 0$.
From equations \eqref{eq:ineq2} and \eqref{eq:gbarp} we have that $\overline{\gamma}^p \geq 1/2$, $\overline{\gamma}^q \geq 1/2$. Then, for all these values of $\overline{\gamma}^q, \overline{\gamma}^p$ it is easy to verify that the factor in front of ${T_{11}}(\overline{\gamma}^q,\overline{\gamma}^p)$ in Eq.~\eqref{eq:dT1dq} is negative. Thus,
\begin{equation}
	\partial T_1/\partial \overline{\gamma}^q < 0, \quad \partial T_1/\partial \overline{\gamma}^p > 0
\end{equation}
and $T_1$ takes its minimal value at the boundary of the allowed region for $\overline{\gamma}^q, \overline{\gamma}^p$, namely, at the point where $\overline{\gamma}^q$ is maximal and $\overline{\gamma}^p$ is minimal. Observe that for ${\mathcal N_2}$ using Eqs.~\eqref{eq:ginq}, \eqref{eq:ineq2} and \eqref{eq:gbarp} we have
\begin{equation}
	\frac{1}{2} < \overline{\gamma}^p < \overline{\gamma}^q < \frac{1}{2}\sqrt{\frac{\gamma_{\rm env}^q}{\gamma_{\rm env}^p}} + \gamma_{\rm env}^q.
\end{equation}
Then $T_1$ takes its minimal value for $\overline{\gamma}^p_{\rm min} = 1/2$ and $\overline{\gamma}^q_{\rm max} = \frac{1}{2}\sqrt{\gamma_{\rm env}^q/\gamma_{\rm env}^p} + \gamma_{\rm env}^q$.
Therefore, if $T_1$ at this point is positive for all values of $\gamma_{\rm env}^q, \gamma_{\rm env}^p$  then it is positive in the whole allowed region of $\overline{\gamma}^q$ and $\overline{\gamma}^p$. 

In order to evaluate $T_1$ at this point we derive it with respect to $\gamma_{\rm env}^q, \gamma_{\rm env}^p$ and we find that 
\begin{equation}
	\begin{split}
		\frac{\partial}{\partial \gamma_{\rm env}^q} \left. T_1 \right|_{\overline{\gamma}^q_{\rm max},\overline{\gamma}^p_{\rm min}} & < 0,\\
		\frac{\partial}{\partial \gamma_{\rm env}^p} \left. T_1 \right|_{\overline{\gamma}^q_{\rm max},\overline{\gamma}^p_{\rm min}} & > 0.
	\end{split}
\end{equation}
Then, again $T_1$ takes its minimal value at the point where $\gamma_{\rm env}^q$ is maximal and $\gamma_{\rm env}^p$ is minimal. At this limit point we find
 \[
	T_1|_{\gamma_{\rm env}^q \rightarrow \infty, \gamma_{\rm env}^p \rightarrow 0} \rightarrow 0,
\] 
where the limit is reached from above. This proves that $T_1 > 0$.

Now we show that $T_2 > 0$ as well. We rewrite $T_2$ as
\begin{equation}\label{eq:T2xnu}
	T_2 = \xi \, \left[(z^2-1)^2 \beta(\nu_{\rm out}) + \frac{8\nu_{\rm out}^2 z}{\nu_{\rm env}}\right]
\end{equation}
where
\begin{equation}
	\begin{split}
		z & = \frac{\gamma_{\rm in}^q}{\gamma_{\rm in (thr)}^{q}}, \quad \xi = \frac{g'\left(\nu_{\rm out} -\frac{1}{2}\right) (\gamma^p_{\rm env})^2}{z^4},\\
		\beta(\nu_{\rm out}) & = \frac{g''\left(\nu_{\rm out} - \frac{1}{2}\right)}{g'\left(\nu_{\rm out} - \frac{1}{2}\right)}\nu_{\rm out} -1,
	\end{split}		
\end{equation}
where $\gamma_{\rm in (thr)}^{q}$ was defined in Eq.~\eqref{eq:ginq} and therefore\\ $0 \leq z < 1$. Using these notations we express
\begin{equation}
	\nu_{\rm out} = \sqrt{1/4 - \nu_{\rm env}/2(z+1/z)+\nu_{\rm env}^2}
\end{equation}
and since $\xi > 0$ we can rewrite the desired inequality $T_2 > 0$ in the equivalent form
\begin{equation}\label{eq:hineq}
	h(z,\nu_{\rm out}) > -\frac{\beta(\nu_{\rm out})}{\nu_{\rm out}^2},
\end{equation}
where 
\begin{equation}
	h(z,\nu_{\rm out}) = \frac{32z^2}{(\sqrt{(1-z^2)^2 + 16z^2\nu_{\rm out}^2} - 1 - z^2)(1-z^2)^2}.
\end{equation}
We observe that
\begin{equation}
	\lim_{z \to 0}h(z,\nu_{\rm out}) = \frac{16}{4\nu_{\rm out}^2 - 1}.
\end{equation}
It is easy to check that the inequality
\begin{equation}
	\frac{16}{4\nu_{\rm out}^2 - 1} > -\frac{\beta(\nu_{\rm out})}{\nu_{\rm out}^2}, 
\end{equation}
holds $\forall \nu_{\rm out} \geq 1/2$. Thus, if $\partial h(z,\nu_{\rm out})/\partial z > 0$ holds for all $\nu_{\rm out}$ and $z$ in the allowed region then the desired inequality \eqref{eq:hineq} holds. We find that
\begin{equation}
	\frac{\partial h(z,\nu_{\rm out})}{\partial z} = 64z \frac{a(z,\nu_{\rm out})}{b(z,\nu_{\rm out})},
\end{equation}
where 
\begin{equation}\label{eq:a}
	\begin{split}
		a(z,\nu_{\rm out}) & = -1 - 8z^2\nu_{\rm out}^2 - 3z^4(8\nu_{\rm out}^2-1) - 2z^6\\
						   & + l(z,\nu_{\rm out})(1 + z^2 + 2z^4),\\
	    b(z,\nu_{\rm out}) & = (z^2-1)^3l(z,\nu_{\rm out})(1+z^2-l(z,\nu_{\rm out}))^2,\\
        l(z,\nu_{\rm out}) & = \sqrt{1+z^4+2z^2(8\nu_{\rm out}^2-1)}.							
	\end{split}	
\end{equation}
Clearly $b(z,\nu_{\rm out})$ is negative and therefore, if $a(z,\nu_{\rm out})$ is negative as well then $\partial h(z,\nu_{\rm out})/\partial z > 0$. Since the first line in $a(z,\nu_{\rm out})$ in Eq.~\eqref{eq:a} is negative in the allowed region of $\nu_{\rm out}$ and $z$, and the second line is positive we can make a comparison of squares of the first and second line, which confirms that indeed $a(z,\nu_{\rm out}) < 0$. Thus, $T_2 > 0$, as well as $T_1 > 0$ which proves \eqref{eq:dFdlambdadinq} which means that \eqref{eq:dgbarpdlambda} holds and thus, the Lemma is proven. \qedhere
\end{proof}

Additionally by combining Eqs.~\eqref{eq:dinqdlambda}, \eqref{eq:dFdlambda} and \eqref{eq:dFdlambdadinq} we conclude that
\begin{equation}
	\frac{d\gamma_{\rm in}^q}{d\lambda} > 0.
\end{equation}
This means that the antisqueezing in the more noisy quadrature is always increasing until the squeezing value at $\lambda = \lambda_{\rm thr}$ is reached.

%%%%%%%%%%%%%%%%%%%%%%%%%%%%%%%%%%%%%%%%%%%%%%%%%%
\subsection{Concavity of the Holevo $\chi$-quantity in $\lambda$}\label{sec:concavity}
%%%%%%%%%%%%%%%%%%%%%%%%%%%%%%%%%%%%%%%%%%%%%%%%%%
\begin{lemholevo} 
	The Holevo $\chi$-quantity given by Eq.~\eqref{eq:chi} is a concave function of $\lambda$, on the solution of the optimization problem.
\end{lemholevo}
\begin{proof}
For $\lambda \geq \lambda_{\rm thr}$ we find that $\gamma_{\rm in}^q$ is given by \eqref{eq:oneoptin} and independent of $\lambda$. Therefore, we conclude from Eq.~\eqref{eq:nuwf} that $\chi$ is on the solution a function of only one variable $\lambda$. Then, at the extremum of ${\mathcal L}$ the second partial derivative of $\chi$ with respect to $\lambda$ is equal to the total second derivative, which reads
\begin{equation}
	\frac{d^2 \chi}{d \lambda^2} = \frac{\partial^2 \chi}{\partial \lambda^2} = \frac{g''(\overline{\nu}_{\rm wf}-\frac{1}{2})}{4} < 0.
\end{equation}
Thus, we have shown that above the threshold $\chi$ is a concave function of $\lambda$.

For an input energy $\lambda$ below the threshold, $\gamma_{\rm in}^q$ depends on $\lambda$ via the implicit function given by Eq.~\eqref{eq:onebt}. Therefore, the total second derivative of $\chi$ with respect to $\lambda$ has to take into account this dependence. Now this reads
\begin{equation}\label{eq:d2chidlambda2}
	\begin{split}
		\frac{d^2 \chi}{d \lambda^2} = & \frac{\partial^2 \chi}{\partial \lambda^2} + \frac{\partial^2 \chi}{\partial \lambda \partial \gamma_{\rm in}^q} \frac{d \gamma_{\rm in}^q}{d \lambda} + \frac{\partial \chi}{\partial \gamma_{\rm in}^q} \frac{d^2 \gamma_{\rm in}^q}{d \lambda^2}\\
		& + \left(\frac{\partial^2 \chi}{\partial \lambda \partial \gamma_{\rm in}^q} + \frac{\partial^2 \chi}{\partial (\gamma_{\rm in}^q)^2} \frac{d \gamma_{\rm in}^q}{d \lambda}\right)\frac{d \gamma_{\rm in}^q}{d \lambda}.
	\end{split} 
\end{equation}
One can easily show using Eq.~\eqref{eq:6eqmodp} that $\partial \chi/\partial \lambda = \mu$ and by Eq.~\eqref{eq:F} it follows that $\partial \chi/\partial \gamma_{\rm in}^q = F$. Thus, Eq.~\eqref{eq:d2chidlambda2} simplifies on the solution to 
\begin{equation}
	\frac{d^2 \chi}{d \lambda^2} = \frac{d \mu}{d \lambda} < 0,
\end{equation}
as proven in Appendix \ref{sec:mulambda}, which proves the lemma. \qedhere
\end{proof}
%%%%%%%%%%%%%%%%%%%%%%%%%%%%%%%%%%%%%%%%%%%%%%%%%%
\section{Eigenvectors and Eigenvalues of Toeplitz matrices}\label{sec:toep}
%%%%%%%%%%%%%%%%%%%%%%%%%%%%%%%%%%%%%%%%%%%%%%%%%%
In the following we derive the optimal input covariance matrix in the case of global water-filling.

All Toeplitz matrices that belong to the Wiener class commute asymptotically, because in this limit they commute with circulant matrices which all commute between each other (using \cite{G06} and \cite{F96}). A circulant matrix $A$ with dimension $n \times n$ is defined as	$A_{ij} = a_{i-j \; {\rm mod} \; n}$.
Therefore, we can introduce the notation $k = (i-j) \; {\rm mod} \; n$ which indicates the $k$th diagonal of $A$. From \cite{G06} we state, that the eigenvalues of $A$ are $\psi_{m} = \sum_{k = 0}^{n-1}{a_k \, e^{-i 2 \pi m k/n}}, \quad m = 1,2,...n$. If we take the limit $n \rightarrow \infty$ the latter becomes the valid solution for the eigenvalue spectrum of all Toeplitz matrices (that belong to the Wiener class). As argued in Sec. \ref{sec:multimode}, the optimal input covariance matrix $\gamma_{\rm in}$ is diagonalized in the same basis as the noise covariance matrix $\gamma_{\rm env}$. Thus, $\gamma_{\rm in}$ is asymptotically Toeplitz with quadrature spectra
\[
	\gamma_{\rm in}^{q,p}(x) = \sum_{k = 0}^\infty{\gamma^{q,p}_{{\rm in},k}\, e^{-i k x}}, \quad x \in [0,2 \pi],
\]
where $\gamma^{q,p}_{{\rm in},k}$ are the $k$th diagonal of $\gamma^{q,p}_{\rm in}$ (and the Fourier coefficient of a Fourier series) in the original basis. Since $\gamma^{q,p}_{\rm in}(x)$ is Riemann integrable we conclude that
\begin{equation}\label{eq:gink}
	\gamma^{q,p}_{{\rm in},k} = \frac{1}{2 \pi}{\int\limits_0^{2\pi}{d x \, e^{i k x}} \, \gamma_{\rm in}^{q,p}(x)}, \quad k = 0,1,2,...,\infty,
\end{equation}
which provides the covariance matrix of the input state in the case of a global water-filling.

%%%%%%%%%%%%%%%%%%%%%%%%%%%%%%%%%%%%%%%%%%%%%%%%%%
\section{Gauss-Markov process of order ${\mathcal P}$}\label{sec:arp}
%%%%%%%%%%%%%%%%%%%%%%%%%%%%%%%%%%%%%%%%%%%%%%%%%%
Here we state the extension of the Gauss-Markov process to a Markov process of order ${\mathcal P}$ which is also called autoregressive (AR) process with white Gaussian noise. The underlying stochastic process is defined as \cite{SZ99}
\begin{equation}\label{eq:arp}
	Z_t = \sum\limits_{k=1}^{\mathcal P}{\phi_k \, Z_{t-k}} + W_t, \quad t=1,...,n
\end{equation}
where $\phi_1,\phi_2,...,\phi_{\mathcal P}$ are the correlation parameters and $W_t$ are identically and independently Gaussian distributed random variables. This process is stationary (shift invariant) iff all roots of the characteristic polynomial
\begin{equation}\label{eq:charpol}
	 p(y) = 1 - \sum\limits_{k=1}^{\mathcal P}{\phi_k \, y^k}
\end{equation} 
lie outside the unit circle $|y|=1$. If the process is stationary, then the covariance matrix of the stochastic process \eqref{eq:arp} is Toeplitz. Its spectrum is given by \cite{SZ99}
\begin{equation}\label{eq:evarp}
	\gamma_{AR}(x) = \frac{\mathrm{Var}(Z_t)}{|1 - \sum_{k=1}^{\mathcal P}{\phi_k \, e^{ikx}}|^2},
\end{equation} 
where $\mathrm{Var}(Z_t)$ is the variance of $Z_t$. Thus, if the noise correlations in both quadrature blocks in \eqref{eq:envcov} are given by an AR process of order ${\mathcal P}$ one immediately can compute the capacity using \eqref{eq:evarp} and the algorithm presented in Sec.~\ref{sec:solarb}.

%%%%%%%%%%%%%%%%%%%%%%%%%%%%%%%%%%%%%%%%%%%%%%%%%%%%%%%%%%%%%%%%%%%%%%%%%%%%%%%%%%%%%%%%%%%%

\begin{thebibliography}{1}

\bibitem{MY04} K.~Matsumoto and F.~Yura,
A:~Math.~Gen. {\bf 37}, L167 (2004).

\bibitem{K03} C.~King, 
IEEE Trans.~Inf.~Theory {\bf 49}, 221 (2003).

% \bibitem{BM04} G.~Bowen and S.~Mancini,
% %Quantum channels with a finite memory
% Phys.~Rev.~A {\bf 69}, 012306 (2004).
% 
% \bibitem{D05} I.~Devetak,
% %The private classical capacity and quantum capacity of a quantum channel
% IEEE Trans. Inf. Theory {\bf 51}, 44 (2005).

%Classical Capacity of the Lossy Bosonic Channel: The Exact Solution
\bibitem{GGLMSY04} V.~Giovannetti, S.~Guha, S.~Lloyd, L.~Maccone, J.~H.~Shapiro and H.~P.~Yuen,
Phys.~Rev.~Lett {\bf 92}, 027902 (2004).

\bibitem{H09} M.~B.~Hastings,
%A Counterexample to Additivity of Minimum Output Entropy
Nat.~Phys. {\bf 5}, 255 (2009).

\bibitem{KW05} D.~Kretschmann and R.~F.~Werner,
%Quantum channels with memory
Phys.~Rev.~A {\bf 72}, 062323 (2005).

% Discrete Channel Capacity Studies
\bibitem{MP02} C.~Macchiavello and G.~M.~Palma,
%Entanglement-enhanced information transmission over a quantum channel with correlated noise
Phys.~Rev.~A {\bf 65}, 050301(R) (2002).

\bibitem{MPV04} C.~Macchiavello, G.~M.~Palma and S.~Virmani,
%Transition behavior in the channel capacity of two-quibit channels with memory
Phys.~Rev.~A {\bf 69}, 010303(R) (2004).

\bibitem{D07} D.~Daems,
%Entanglement-enhanced transmission of classical information in Pauli channels with memory: Exact solution
Phys.~Rev.~A {\bf 76}, 012310 (2007).

\bibitem{KDC06} E.~Karpov, D.~Daems and N.~J.~Cerf,
%Entanglement enhanced classical capacity of quantum communication channels with correlated noise in arbitrary dimensions
Phys.~Rev.~A {\bf 74}, 032320 (2006).

\bibitem{KM06} V.~Karimipour and L.~Memarzadeh,
%Transition behavior in the capacity of correlated noisy channels in arbitrary dimensions
Phys.~Rev.~A {\bf 74}, 032332 (2006).

\bibitem{BDM05} G.~Bowen, I.~Devetak and S.~Mancini,
%Bounds on classical information capacities for a class of quantum memory channels
Phys.~Rev.~A {\bf 71}, 034310 (2005).

%\bibitem{PV07} M.~B.~Plenio and S.~Virmani,
%Phys.~Rev.~Lett. {\bf 99}, 120504 (2007).

%\bibitem{ABF07} A.~D'Arrigo, G.~Benenti and G.~Falci,
%New.~J.~Phys {\bf 9}, 310 (2007).

%%%%%%%%%%%%%%%%%%%%%%%%%%%%%%
% Continuous...
\bibitem{CCMR05} N.~J.~Cerf, J.~Clavareau, C.~Macchiavello and J.~Roland,
%Quantum entanglement enhances the capacity of bosonic channels with memory
Phys.~Rev.~A {\bf 72}, 042330 (2005).

\bibitem{GM05} V.~Giovannetti and S.~Mancini,
%Bosonic memory channels
Phys.~Rev.~A {\bf 71}, 062304 (2005).

% \bibitem{RSGM05} G.~Ruggeri, G.~Soliani, V.~Giovannetti and S.~Mancini,
% %Information transmission through lossy bosonic memory channels
% Europhys. Lett. {\bf 70}, 719 (2005).

\bibitem{PZM08} O.~V.~Pilyavets, V.~G.~Zborovskii and S.~Mancini,
%Lossy bosonic quantum channel with non-Markovian memory
Phys.~Rev.~A {\bf 77}, 052324 (2008).

\bibitem{LPM09} C. Lupo, O.~V.~Pilyavets and S. Mancini,
%Capacities of lossy bosonic channel with correlated noise
New J. Phys. {\bf 11}, 063023 (2009).

\bibitem{SDKC09} J.~Sch\"afer, D.~Daems, E.~Karpov and N.~J.~Cerf,
Phys.~Rev.~A {\bf 80}, 062313 (2009).

\bibitem{PLM09} O.~V.~Pilyavets, C.~Lupo and S.~Mancini,
%Methods for Estimating Capacities of Gaussian Quantum Channels
arXiv:0907.1532v2 [quant-ph].

\bibitem{LGM10} C.~Lupo, V.~Giovannetti and S.~Mancini,
%Capacities of Lossy Bosonic Memory Channels
Phys. Rev. Lett {\bf 104}, 030501 (2010).

\bibitem{LM10} C.~Lupo and S.~Mancini,
%On the transitional behavior of quantum Gaussain memory channels
Phys. Rev. A {\bf 81}, 052314 (2010).

\bibitem{CD94} C. M. Caves and P. D. Drummond
%Quantum limits on bosonic communication rates,
Rev.~Mod.~Phys. {\bf 66}, 481 (1994).

\bibitem{HSH99} A.~S.~Holevo, M.~Sohma, O.~Hirota,
%\emph{Capacity of quantum Gaussian channels},
Phys.~Rev.~A {\bf 59}, 1820 (1999).

\bibitem{H05} T.~Hiroshima,
%\emph{Additivity and multiplicatitvity properties of some Gaussian channels
%for Gaussian inputs},
Phys. Rev.~A {\bf 73}, 012330 (2006).

\bibitem{LFH81} J.~S.~Lee and R.~H.~French. and Y.~K.~Hong
%Error performance of differentially coherent detection of binary DPSK data transmission on the hard-limiting satellite channel
IEEE Trans. Inf. Theory {\bf 27}, 489 (1981).

\bibitem{Cover} \textit{Elements of Information Theory}, 
T.~M.~Cover and J.~A.~Thomas, (John Wiley \& Sons, Inc., New York, 1991).

\bibitem{H73} A.~S.~Holevo,
Probl. Peredachi Inf. {\bf 9}, 3--11 (1973).
%Bounds for the Quantity of Information Transmitted by a Quantum Communication Channel

\bibitem{SW97} B.~Schumacher and M.~D.~Westmoreland,
%Sending classical information via noisy quantum channels
Phys.~Rev.~A {\bf 56}, 131 (1997).

\bibitem{H98} A.~S.~Holevo,
IEEE Trans. Inf. Theory {\bf 44}, 269 (1998).

\bibitem{GLMSY04} V.~Giovannetti, S.~Lloyd, L.~Maccone, J.~H.~Shapiro and B.~J.~Yen,
%Minimum R\´enyi and Wehrl entropies at the output of bosonic channels
Phys.~Rev.~A {\bf 70}, 022328 (2004).

\bibitem{L09} S.~Lloyd \emph{et. al.}, arXiv:0906.2758v3 [quant-ph]
%The bosonic minimum output entropy conjecture and Lagrangian minimization

\bibitem{SEW05} A.~Serafini, J.~Eisert and M.~M.~Wolf,
Phys.~Rev.~A {\bf 71}, 012320 (2005).

\bibitem{S00} S.~M.~Stefanov,
%Convex Seperable Minimization Subject to Bounded Variables
Comput. Optim. Appl. {\bf 18}, 27-48 (2001).

\bibitem{SZ99} H.~von Storch, F.~W.~Zwiers,
\emph{Statistical Analysis in Climate Research}, Cambridge University Press, Cambridge (1999).

\bibitem{LMM09} C. Lupo, L. Memarzadeh and S. Mancini,
Phys.~Rev.~A {\bf 80}, 042328 (2009).

\bibitem{G06} R.~M.~Gray, 
\emph{Foundations and Trends in Communications and Information Theory} {\bf 2}, 155 (2006).

\bibitem{F96} W.~A.~Fuller, 
\emph{Introduction to Statistical Time Series}, Wiley, New York, (1996).


\end{thebibliography}
\end{document}